\documentclass[a4paper]{article}
\usepackage[utf8]{inputenc}
\usepackage[english]{babel}
\usepackage{bm}
\usepackage[a4paper,top=3cm,bottom=2cm,left=2cm,right=2cm,marginparwidth=1.75cm]{geometry}
\usepackage{multirow}
\usepackage{authblk}
\usepackage{url}
\usepackage{color}
\usepackage{amsmath}
\usepackage{amsthm}
\usepackage[round,comma]{natbib}
\usepackage{graphicx}
\usepackage{tikz}
\usepackage[boxed]{algorithm2e}
\usepackage{subcaption}
\usepackage{setspace}
\usepackage{caption}
\makeatletter
\newtheorem*{rep@theorem}{\rep@title}
\newcommand{\newreptheorem}[2]{%
\newenvironment{rep#1}[1]{%
 \def\rep@title{#2 \ref{##1}}%
 \begin{rep@theorem}}%
 {\end{rep@theorem}}}
\makeatother

\newreptheorem{lemma}{Lemma}
\newtheorem{proposition}{Proposition}
\newreptheorem{proposition}{Proposition}

\begin{document}

\title{ \vspace{-3em}  Detecting Heterogeneous Treatment Effect with Instrumental Variables}
\author[1]{Michael Johnson}
\author[1]{Jiongyi Cao}
\author[1]{Hyunseung Kang}
\affil[1]{Department of Statistics, University of Wisconsin-Madison}
\date{}

\setcounter{Maxaffil}{0}
\renewcommand\Affilfont{\itshape\small}

\maketitle

\begin{abstract}
\noindent There is an increasing interest in estimating heterogeneity in causal effects in randomized and observational studies. However, little research has been conducted to understand effect heterogeneity in an instrumental variables study. In this work, we present a method to estimate heterogeneous causal effects using an instrumental variable with matching. The method has two parts. The first part uses subject-matter knowledge and interpretable machine learning techniques, such as classification and regression trees, to discover potential effect modifiers. The second part uses closed testing to test for statistical significance of each effect modifier while strongly controlling the familywise error rate. We apply this method on the Oregon Health Insurance Experiment, estimating the effect of Medicaid on the number of days an individual's health does not impede their usual activities by using a randomized lottery as an instrument. Our method revealed Medicaid's effect was most impactful among older, English-speaking, non-Asian males and younger, English-speaking individuals with at most high school diplomas or General Educational Developments.

\noindent {\bf Keywords:} Complier Average Causal Effect, Heterogeneous Treatment, Instrumental Variables, Matching, Machine Learning, Oregon Health Insurance Experiment
\end{abstract}

\section{Introduction}
\subsection{Motivation: Utilization of Medicaid in Oregon and the Complier Average Treatment Effect}

In January of 2008, Oregon reopened its Medicaid-based health insurance plan for its eligible residents and, for a brief period, allowed a limited number of individuals  to enroll in the program. Specifically, a household in Oregon was randomly selected by a lottery system run by the state and any eligible individual in the household can choose to enroll in the new health insurance plan; households that weren't selected by the lottery could not enroll whatsoever. 

For policymakers, Oregon's randomized lottery system was a unique opportunity, specifically a natural experiment, to study Medicaid's causal effect on a variety of health and economic outcomes, as directly randomizing Medicaid (or withholding it) to individuals would be infeasible and unethical. In this natural experiment, commonly referred to as the Oregon Health Insurance Experiment (OHIE), \citet{finkelstein2012oregon} used the randomized lottery as an instrumental variable (see Section \ref{ivcate} for details) to study the complier average causal effect (CACE), or the effect of Medicaid among individuals who enrolled in Medicaid after winning the lottery \citep*{angrist1996identification}. The CACE reflects Medicaid's impact among a subgroup of individuals and differs from the average treatment effect for the entire population (ATE) or the intent-to-treat (ITT) effect of the lottery itself on the outcome. In this paper, we focus on studying the CACE; see \citet{imbens2010better}, \citet{swanson2013commentary}, and \citet{swanson2014think} for additional discussions on the CACE.

Often in studying the CACE, the population of compliers is assumed to be homogeneous whereby two compliers are alike and have the same treatment effect. But, no two individuals are the same and it is plausible that some compliers may benefit more from the treatment than other compliers. For example, sick individuals who enroll in Medicaid after winning the lottery may benefit more from Medicaid than healthy individuals. Also, the perceived benefit of enrolling in Medicaid among sick versus healthy individuals may create heterogeneity in the compliance rate, i.e. the number of people who sign up when they win the lottery, with sick people presumably signing up more than healthy people. Alternatively, if people are equally likely to enroll in Medicaid when they win the lottery, those who are unemployed may benefit more from Medicaid in terms of reducing out-of-pocket healthcare spending and medical debt than those who are employed. The theme of this paper is to explore these issues, specifically the heterogeneity of CACE and how to discover them in an honest manner by using well-known matching methods and recent tree-based methods in heterogeneous treatment effect estimation. 

\subsection{Prior Work and Our Contributions}
There are many recent works in causal inference using tree-based methods to estimate heterogeneity in the ATE, with majority of them utilizing sample splitting or sub-sampling to obtain honest inference; see \citet{su2009subgroup}, \citet{hill2011bayesian}, \citet{athey2016recursive}, \citet{hahn2017bayesian}, \citet{wager2018estimation}, \citet{chernozhukov2018generic}, \citet*{athey2019generalized} and references therein. Here, honest inference refers to a procedure that controls the Type I error rate (or the familywise error rate) of testing a null hypothesis about a treatment effect at a desired level $\alpha$, especially if the hypothesis was suggested by the data; see Section \ref{dischcate} for additional discussions. \citet*{hsu2013effect} used pair matching and classification and regression trees (CART) \citep{breiman1984classification} to conduct honest inference, all without sample splitting. A follow-up work by \citet{hsu2015strong} formally showed that the procedure strongly controls the familywise error rate for testing heterogeneous treatment effects. Subsequent works by \cite*{lee2018air},  \cite{lee2018surgery}, and \cite{lee2018powerful} extended this idea to increase statistical power of detecting such effects. 

There is also work on nonparametrically estimating treatment effects using instrumental variables (IV), mostly using likelihood, series, sieve, minimum distance, and/or moment-based methods \citep{abadie2003semiparametric, blundell2003endogeneity, newey2003instrumental, ai2003efficient, hall2005nonparametric, blundell2007semi, darolles2011nonparametric, chen2012estimation, su2013local,athey2019generalized}. Recently, \cite{stoffi2018estimating} and  \citet{bargagli2019heterogeneous} explored effect heterogeneity in the CACE by using causal trees \citep{athey2015machine} and Bayesian causal forests \citep{hahn2017bayesian}, specifically by estimating heterogeneity in the ITT effect and dividing it by the compliance rate. However, to the best of our knowledge, none of the methods used matching, a popular, intuitive, and easy-to-understand method in causal inference, as a device to nonparametrically estimate treatment heterogeneity in the CACE and to guarantee strong familywise Type I error control. Works on using matching with an instrument by \citet{baiocchi2010building} and \cite{kang2013causal, kang2016full} only focused on the population CACE; they do not explore heterogeneity in the CACE. Also, aforementioned works by \cite{hsu2013effect} and \cite{hsu2015strong} using matching and CART did not consider instruments. 

The goal of this paper is to propose a matching-based method to study effect heterogeneity in instrumental variables settings. Specifically, the target estimand of interest is what we call the \emph{heterogeneous} complier average causal effect (H-CACE). A heterogeneous complier average causal effect (H-CACE) is the usual complier average causal effect, but for a subgroup of individuals defined by their pre-instrument covariates. At a high level, H-CACE explores  treatment heterogeneity in the complier population, where we suspect that not all compliers in the data react to the treatment in the same way. Some subgroup of compliers may respond to the treatment differently than another subgroup of compliers, who may not respond to the treatment at all; some may even be more likely to be compliers if they believe the treatment would benefit them and they may actually benefit from the treatment. The usual CACE obscures the underlying heterogeneity among compliers by averaging across different types of compliers whereas H-CACE attempts to expose it. Also, in the case where the four compliance types in \cite{angrist1996identification}, specifically compliers, never-takers, always-takers, and defiers, have identical effects, the H-CACE can identify the heterogeneous treatment effect for the entire population using an instrument. Section \ref{hcate} formalizes H-CACE and provides additional discussions.

Methodologically, to study H-CACE, we combine existing ideas of heterogeneous treatment effect estimation in non-IV matching contexts by \citet{hsu2015strong} and matching with IVs by \cite{baiocchi2010building} and \cite{kang2016full}. Specifically, we first follow \cite{baiocchi2010building} and \cite{kang2016full} and conduct pair matching on a set of pre-instrument covariates. Second, we follow \cite{hsu2015strong} where we obscure the difference in the outcomes between treated and controls by using absolute differences and use CART to discover novel subgroups of study units without contaminating downstream inference. Specifically, we use closed testing to test the H-CACE in different subgroups while strongly controlling for familywise error rate \citep*{marcus1976closed}. Simulation studies are conducted to evaluate the performance of our proposed method under varying levels of compliance and effect heterogeneity. The simulation study also compares our method to the recent aforementioned method by \cite{bargagli2019heterogeneous}. We then use our method to analyze heterogeneity in the effect of Medicaid on increasing the number of days a complying individual's health does not hamper their usual activities.

\section{Method}
\subsection{Notation}

Let $i=1,\dots, I$ index the $I$ matched pairs and $j=1,2$ index the units within each matched pair $i$. Let $Z_{ij}$ be a binary instrument for unit $j$ in matched pair $i$ where one unit in the pair receives the instrument value $Z_{ij}=1$ and the other receives the value $Z_{ij} = 0$. In the OHIE data, $Z_{ij} = 1$ and $Z_{ij}=0$ denotes an individual winning or losing the Medicaid lottery, respectively. Let $\mathbf{Z}$ be the vector of instruments, $\mathbf{Z} = (Z_{11},Z_{12}, \dots, Z_{I1}, Z_{I2}$) and $\mathcal{Z}$ denote an event of instrument assignments for all units.

For unit $j$ in matched pair $i$, let $d_{1ij}$ and $d_{0ij}$ denote the binary potential treatment/exposure given the instrument value of $Z_{ij}=1$ and $Z_{ij}=0$ respectively. Further, define the potential response $r_{1ij}^{(d_{1ij})}$ for unit $j$ in matched set $i$ with exposure $d_{1ij}$ receiving instrument value $Z_{ij} = 1$; we define $r_{0ij}^{(d_{0ij})}$ similarly but with instrument value $Z_{ij} = 0$. For the OHIE data, $d_{1ij}$ denotes whether an individual enrolled in Medicaid and $r_{1ij}^{(d_{1ij})}$ denotes the potential outcome when the individual wins the lottery $Z_{ij} = 1$. For unit $j$ in matched set $i$, the observed response is defined as $R_{ij} = r_{1ij}^{(d_{1ij})} Z_{ij} + r_{0ij}^{(d_{0ij})}(1-Z_{ij})$ and the observed treatment is defined as $D_{ij} = d_{1ij}Z_{ij} + d_{0ij}(1-Z_{ij})$. The notation assumes that the Stable Unit Treatment Value Assumption (SUTVA) holds \citep{rubin1980randomization}. Define $\mathcal{F} = \{ (r_{1ij}^{(d_{1ij})}, r_{0ij}^{(d_{0ij})}, d_{1ij}, d_{0ij}, \mathbf{X}_{ij}, u_{ij}), i=1,\dots,I, j=1,2\}$ to be the set of potential outcomes, treatments, and covariates, both observed, $\mathbf{X}_{ij}$, and unobserved, $u_{ij}$. 

When partitioning the matched sets into subgroups for discovering effect heterogeneity, the following notation is used. We define a ``set of sets", or grouping $\mathcal{G}$, which contains mutually exclusive and exhaustive subsets of the pairs $s_g \subseteq \{ 1, \dots , I\}$ so that $\mathcal{G} = \{ s_1, \dots, s_G\}$. The subscript $g$ in $s_g$ is used to denote a unit partitioned into the $g$th subset $s_g$. To avoid overloading the notation, $s$ and $s_g$ will be used interchangeably when it isn't necessary to specify a subgroup $g$. The set of potential outcomes, treatments, and covariates for subset $s_g$ are defined as $\mathcal{F}_{s_g} = \{(r_{1sij}^{(d_{1sij})},r_{0sij}^{(d_{0sij)}}, d_{1sij},d_{0sij}, \mathbf{X}_{sij}, u_{sij}): s_g  \subseteq \{ 1, \dots , I\}, i \in s_g, j = 1, 2\}$, where $\mathcal{F} = \bigcup_s \mathcal{F}_s$. For example, consider a grouping of two subgroups, $\mathcal{G} = \{s_1, s_2 \}$, for $I=10$ matched pairs. Suppose the first few pairs and the last pair make up the first subgroup and the rest are in the second subgroup, say $s_1 = \{1, 2, 3, 10\}$ and $s_2 = \{4, 5, 6, 7,8 ,9\}$. The set of potential responses, treatments, and covariates for the first group is then $\mathcal{F}_{s_1} = \{(r_{1s_1ij}^{(d_{1s_1ij})},r_{0s_1ij}^{(d_{0s_1ij)}}, d_{1s_1ij},d_{0s_1ij}, \mathbf{X}_{s_1ij}, u_{s_1ij}): s_1  = \{1, 2, 3, 10\}, i \in s_1, j = 1, 2\}$. The observed response, binary instrument, and exposure for a given unit in subset $s_g$ is  denoted as $Z_{s_gij}$, $R_{s_gij}$, and $D_{s_gij}$ respectively.

\subsection{Review: Matching, Instrumental Variables, and the CACE} \label{ivcate}

Matching is a popular non-parametric technique in observational studies to balance the distribution of the observed covariates between treated and control units by grouping units based on the similarity of their covariates; see \citet{stuart2010matching}, Chapters 3 and 8 of \citet{rosenbaum2010design}, and \cite{rosenbaum2020modern} for overviews of matching. Pair matching is a specific type of matching where each treated unit is only matched to one control unit. In the context of instrumental variables and pair matching, the instrument serves as the treatment/control variable and the matching algorithm creates $I$ matched pairs where the two units in a matched pair are similar in their observed covariates $x_{ij}$, but one receives the instrument value $Z_{ij} = 1$ and the other receives the instrument value $Z_{ij} = 0$. 

Instrumental variables (IV) is a popular approach to analyze causal effects when unmeasured confounding is present and is based on using a variable called an instrument \citep{angrist1996identification,hernan2006instruments,baiocchi2014instrumental}. The instrument must satisfy three core assumptions: (A1) the instrument is related to the exposure or treatment, or $\sum_{i=1}^I \sum_{j=1}^2 (d_{1ij}-d_{0ij}) \neq 0$ (commonly referred to as instrument relevance); (A2) the instrument is not related to the outcome in any way except through the treatment, or $r_{0ij}^{(d)} = r_{1ij}^{(d)} \equiv r_{ij}^{(d)}$ for a fixed $d$ (commonly referred to as the exclusion restriction); and (A3) the instrument is not related to any unmeasured confounders that affect the treatment and the outcome, or $P(Z_{ij} = 1|\mathcal{F}, \mathcal{Z}) = \frac{1}{2}$ within each pair $i$ (commonly referred to as instrument ignorability or exchangeability). If these core assumptions are satisfied, it is possible to obtain bounds on the average treatment effect \citep{balke1997bounds}. To point identify a treatment effect, one needs to make additional assumptions. Here, we assume (A4) monotonicity where the potential treatment is a monotonic function of the instrument values, or $d_{0ij} \leq d_{1ij}$. Assumption (A4) can be interpreted in terms of four sub-populations: compliers, always-takers, never-takers, and defiers \citep{angrist1996identification}. Compliers are units which their treatment values follow their instrument values, or $d_{0ij} = 0, d_{1ij} = 1$. Always-takers always take the treatment regardless of their instrument values, or $d_{0ij} = d_{1ij}=1$. Never-takers never take the treatment regardless of their instrument values, or $d_{0ij} = d_{1ij}=0$. Defiers act against their instrument values, or $d_{0ij}=1, d_{1ij} = 0$. Assumption (A4) then states that no defiers exist.

Let $N_{\rm CO}$ be the total number of compliers in the population. Under the IV assumptions (A1)-(A4),  the CACE, formally defined as
\begin{equation} \label{effectratio}
   \lambda = \frac{\sum_{i=1}^{I} (r_{1ij}^{(1)}-r_{0ij}^{(0)} ) I(d_{1ij} = 1, d_{0ij} = 0)}{\sum_{i=1}^I \sum_{j=1}^{2} d_{1ij} - d_{0ij} } = \frac{1}{N_{\rm CO}} \sum_{i=1}^{I} (r_{1ij}^{(1)}-r_{0ij}^{(0)} ) I(ij \text{ is a complier})
\end{equation}
can be identified from data by taking the ratio of the estimated ITT effect over the estimated compliance rate. In the context of matching and instrumental variables, \citet{baiocchi2010building} and \citet{kang2016full}  proposed a test statistic to test the null $H_0: \lambda = \lambda_0$ by using differences in the adjusted outcomes
\begin{equation}\label{teststat}
T(\lambda_0) = \frac{2}{I}\sum_{i=1}^I \sum_{j=1}^2 Z_{ij}(R_{ij} - \lambda_0 D_{ij}) - (1-Z_{ij})(R_{ij} - \lambda_0 D_{ij})
\end{equation}
along with an estimator for the variance of $T(\lambda_0)$,
\begin{equation}\label{teststatvar}
S^2(\lambda_0) = \frac{1}{I(I-1)} \sum_{i=1}^I \sum_{j=1}^2 \left(Z_{ij}(R_{ij} - \lambda_0 D_{ij}) - (1-Z_{ij})(R_{ij} - \lambda_0 D_{ij}) - T(\lambda_0) \right)^2
\end{equation}
Under the null, \cite{baiocchi2010building} and \citet{kang2016full} showed that $\frac{T(\lambda_0)}{S(\lambda_0)}$ asymptotically follows a standard Normal distribution. For point estimation, the same set of authors proposed  a Hodges-Lehmann type estimator \citep{hodges1963estimates} which involves solving $\lambda$ in the equation $T(\lambda)/S(\lambda) = 0$. For a $1-\alpha$ \% confidence interval, the equation $T(\lambda)/S(\lambda) \leq z_{1-\alpha/2}$ is solved for $\lambda$, where $z_{1-\alpha/2}$ is the $1-\alpha/2$ quantile of the standard Normal distribution; see \cite{kang2016full} and \cite{kang2018inference} for details.

\subsection{Heterogeneous Complier Average Causal Effect (H-CACE)} \label{hcate}
We formally define the target estimand of interest in the paper, the heterogeneous treatment effect among compliers, or H-CACE. Formally, the H-CACE is defined as the CACE for a subgroup of compliers with a specific value of covariates 
\begin{equation} \label{hcateEqn}
\lambda(\mathbf{x}) = \frac{ \sum_{i=1}^{I} \sum_{j=1}^{2} (r_{1ij}^{(1)} - r_{0ij}^{(0)})I(d_{1ij} = 1, d_{0ij} = 0, \mathbf{X}_{ij} = \mathbf{x}) }{ \sum_{i=1}^{I} \sum_{j=1}^{2} (d_{1ij} - d_{0ij}) I(\mathbf{X}_{ij} = \mathbf{x}) }
\end{equation}
Because two units are assumed to have identical covariate values within each matched pair, $\lambda(\mathbf{x})$ can be rewritten as taking a subset of $I$ matched pairs with identical covariates $\mathbf{x}$, say $s \subseteq \{1,\ldots,I\}$ 
\begin{equation} \label{subseteffectratio}
     \lambda_s = \frac{ \sum_{i \in s} \sum_{j=1}^{2} r_{1sij}^{(d_{1sij})} - r_{0sij}^{(d_{0sij})} }{ \sum_{i \in s} \sum_{j=1}^{2} d_{1sij} - d_{0sij}}
\end{equation}
Since each H-CACE $\lambda_s$ has the same form as the original CACE, we can apply the test statistic in Section \ref{ivcate}. Formally, consider the subset-specific hypothesis $H_{0s}: \lambda_s = \lambda_0$ against $H_{1s}: \lambda_s \neq \lambda_0$. Then, we can use the test statistic (\ref{teststat}) with variance (\ref{teststatvar}) among the pairs specific to subset $s$.

Also, under assumptions (A1)-(A4), for a mutually exclusive and exhaustive grouping $\mathcal{G} = \{s_1, \dots, s_G\}$ of a set of pairs $s_g \subseteq \{1, \dots, I\}$ with at least one complier within each subgroup $s_g$, the original CACE is equal to a weighted version of H-CACE: 
\[
\lambda = \sum_{g=1}^G w_{s_g}\lambda_{s_g}, \quad{}
w_{s_g} = \frac{\sum_{i \in s_g} \sum_{j=1}^2 d_{1sij} - d_{0sij}}{N_{\rm CO}}.
\]
An implication of this expression is that typical analysis of the CACE hides underlying effect heterogeneity. For example, suppose there are two subgroups defined by a binary covariate, say male or female, and consider two scenarios. In the first scenario, among compliers, 80\% are male and 20\% are female. Also, the H-CACE of male is 1.25 and the H-CACE of female is 0. In the second scenario, the male/female complier proportions remain the same, but the H-CACE of male is now 1.5 and the H-CACE of female is -1. In both scenarios, the CACE is $1$. But, in the second scenario, females have a negative treatment effect. By only studying the CACE, as is typical in practice, variations in the treatment effects defined by H-CACEs would have been masked. The next section presents a way to unwrap the CACE and discover novel H-CACEs. 

\subsection{Discovering and Testing Novel H-CACE}\label{dischcate}
A naive approach to finding and testing novel H-CACE would be to exhaustively test every H-CACE for every subset of matched pairs and gradually aggregate them based on their covariate similarities with appropriate statistical tests. However, this procedure will not only lead to false discoveries, but it will also be grossly underpowered. 

Instead, based on the work by \citet{hsu2015strong}, we propose to use exploratory machine learning methods, such as CART, to discover and aggregate matched pair into subgroups with similar treatment effects, formulating grouping $\mathcal{G}$. We will then use closed testing to test effect heterogeneity defined by these groups while strongly controlling the familywise error rate; see Algorithm \ref{genprocedure} for details. 

\IncMargin{1em}
\begin{algorithm}[ht]
\DontPrintSemicolon
\SetKwInOut{Input}{Given}\SetKwInOut{Output}{Output}
\SetKwBlock{NewBlock}{}{}
\Input{Observed outcome $R$, binary instrument $Z$, exposure $D$, covariates $X$, null value $\lambda_0$ for testing, and desired familywise error rate $\alpha$}
\BlankLine
\nl Pair match on observed covariates.\;
\nl Calculate absolute value of pairwise differences for each matched pair
\[\bigl\lvert Y_i \bigr\rvert = \bigl\lvert(Z_{i1} - Z_{i2})(R_{i1} - \lambda_0 D_{i1} - (R_{i2} - \lambda_0 D_{i2})) \bigr\rvert\]\;
\vspace{-4mm}
\nl Construct mutually exclusive and exhaustive grouping 
using CART. Here, CART takes $|Y_i|$ as the outcome and $\mathbf{X}_i$ from each matched pair as the predictors. CART outputs a partition of covariates, which we use to define $\mathcal{G}= \{s_1, \dots, s_G\}$ and consequently, H-CACEs. \;
\nl Run closed testing \citep{marcus1976closed} to test statistical significance of H-CACEs for every subset  $\mathcal{L} \subseteq \{1, \dots, G\}$ of $G$ groups where each subset defines the null hypothesis of the form $H_{0 \mathcal{L}}: r_{1ij}^{(d_{1ij})} - r_{0ij}^{(d_{0ij})} = \lambda_0 (d_{1ij} - d_{0ij})$ for all $g \in \mathcal{L}$. Formally, run \;
\Indp \vspace{-3mm}

\Indm\SetAlgoNoLine\NewBlock{
    \SetAlgoLined\For{$ \mathcal{L} \subseteq \{1, \dots, G \} $}{ 
        \If{$H_{0\mathcal{L}}$ \emph{has not been accepted}}{
        Calculate $T_{s}(\lambda_0)$ and $S_{s}(\lambda_0)$ for $s=\bigcup_{g \in \mathcal{L}} s_g$\;
        \If{$\bigl\lvert\frac{T_{s}(\lambda_0)}{S_{s}(\lambda_0)}\bigr\rvert \leq z_{1-\alpha/2}$}{
        Accept the null hypothesis $H_{0\mathcal{K}}: \lambda_{\mathcal{K}} = \lambda_0$ for all $\mathcal{K} \subseteq \mathcal{L} \subseteq \{1, \dots, \mathcal{G}\}$}
        \Else{Reject $H_{0 \mathcal{L}}$}
        }
    }
}
\Output{Estimated and inferential quantities for H-CACEs (e.g. effect size, confidence interval, $p$-value) and novel H-CACEs from closed testing.}
\caption{Proposed method to discover and test effect heterogeneity in IV with matching}\label{genprocedure}
\end{algorithm}\DecMargin{1em}

We explain in some detail the key steps in Algorithm \ref{genprocedure}. First, the specification of the null value $\lambda_0$ is for testing the sharp null of the form $H_0: r_{1ij}^{(d_{1ij})} -  r_{0ij}^{(d_{0ij})}  =\lambda_0(d_{1ij}  - d_{0ij}$); this sharp null implies the ``weak'' or composite null $H_0: \lambda = \lambda_0$ \citep{baiocchi2010building}. Setting $\lambda_0 = 0$ would test whether the H-CACE is zero or not and is the typical choice in most applications unless other null values are of scientific interest. Second, under the sharp null, the absolute value of the difference in adjusted outcomes between pairs, $\bigl\lvert Y_i \bigr\rvert = \bigl\lvert(Z_{i1} - Z_{i2})(R_{i1} - \lambda_0 D_{i1} - (R_{i2} - \lambda_0 D_{i2})) \bigr\rvert$, obscures the instrument assignment vector making $\bigl\lvert Y_i \bigr\rvert $ a function of $\mathcal{F}$ only, a fixed (and unknown) quantity. In contrast, $Y_i$ is a function of both $\mathcal{F}$ and $\mathbf{Z}$. Consequently, conditional on $\mathcal{F}$, building a CART tree based on $|Y_i|$ as the response and $\mathbf{X}_{i}$ as the explanatory variables does not affect the distribution of $\mathbf{Z}$. The distribution of $\mathbf{Z}$ within each pair remains $1/2$ as stated in assumption (A3) and is a key ingredient to achieve familywise error rate control for downstream inference; see our discussion on honest inference below. 

Third, Algorithm \ref{genprocedure} applies closed testing, a multiple inference procedure by \cite{marcus1976closed}, to test for multiple hypotheses about H-CACEs generated by CART's grouping $\mathcal{G} = \{s_1, \dots, s_G\}$. Broadly speaking, closed testing will test sharp null hypotheses defined by every parent and child node of the estimated tree from CART and reject/accept these hypotheses while controlling for multiple testing issues; see Section \ref{sec:app_closedtest} and Figure \ref{fig:ctex} for visualizations. A bit more formally, for each pair, closed testing will test the global sharp null hypothesis $H_{0}: r_{1ij}^{(d_{1ij})} - r_{0ij}^{(d_{0ij})} = \lambda_0 (d_{1ij} - d_{0ij})$ and  subsequent subset-specific hypotheses $H_{0 \mathcal{L}}: r_{1ij}^{(d_{1ij})} - r_{0ij}^{(d_{0ij})} = \lambda_0 (d_{1ij} - d_{0ij})$ for all $g \in \mathcal{L}$, where $\mathcal{L}$ is a subset of the $G$ groups formed by CART. We note that the difference between the global null and the subset-specific nulls is only in the pairs under consideration; all the nulls use the test statistics introduced in Section \ref{ivcate}. Also, the subset-specific hypotheses imply $H_{0 \mathcal{L}}: \lambda_s=\lambda_0$ for $s=\bigcup_{g \in \mathcal{L}} s_g$. 
Closed testing would only reject the subset-specific hypotheses $H_{0 \mathcal{L}}$ if all of the $p$-values from superset hypotheses $H_{0 \mathcal{L}'}$, $\mathcal{L} \subseteq \mathcal{L}'$, are less than $\alpha$. 

As mentioned earlier, the key step of using $|Y_i|$ in CART allows for both discovery and downstream honest testing of H-CACEs via closed testing; again, honesty refers to control of the familywise error rate at level $\alpha$ when testing multiple hypotheses about H-CACEs that were discovered by data. Because $|Y_i|$ is not a function of $\mathbf{Z}$, the original distribution of $\mathbf{Z}$ is preserved and  we can use the standard randomization inference null distribution to honestly test each H-CACE discovered by CART. In fact, as noted in \cite{hsu2015strong}, this honesty property is preserved for any supervised machine learning algorithm that forms groups based on $\mathbf{X}$ and $|Y|$ as well as subsequent visual heuristics to check the algorithms' performance.  Also, in recent work on estimating heterogeneous causal effects \citep{chernozhukov2018generic,athey2019generalized,park2020groupwise}, the notion of "honest" inference is often tied to sample splitting, where one subsample is used to discover different subgroups or to estimate nuisance parameters and the other subgroup is used to test the causal effect. Our approach does not have to use sample splitting to obtain honest inference and Proposition \ref{algorithmprop} shows this principle formally; Web Appendix A shows this principle numerically.

\begin{proposition}[Familywise Error Rate Control of Algorithm \ref{genprocedure}] \label{algorithmprop}
Under the sharp null hypotheses $H_{0\mathcal{L}}$ in Algorithm \ref{genprocedure}, the conditional probability given $(\mathcal{F}, \mathcal{Z}, \mathcal{G})$ that the algorithm makes at least one false rejection of the set of hypotheses is at most $\alpha$.
\end{proposition}

We now discuss some important limitations of Proposition \ref{algorithmprop} and the proposed algorithm. First, our algorithm's guarantee on controlling the familywise error rate is only for testing sharp nulls. As noted in Section 2, page 289 of \cite{rosenbaum2002covariance}, testing for sharp nulls does not necessarily imply that the true data generating process always follow the sharp null and as such, the proposition makes no claims about how the true data generating process actually looks like. Having said that, the limitation of testing a sharp null versus a weak null has been discussed extensively; see Sections 3 and 4 of \cite{rosenbaum2002covariancerejoinder}, \cite{ding2017paradox}, \cite{fogarty2018regression}, and \cite{fogarty2020studentized}. But, a recent work by \cite{fogarty2020biased} has shown that testing the sharp null based on our test is an asymptotically valid test for the weak null; see Remark 1 of their Proposition 1. This suggests that the guarantees from Proposition \ref{genprocedure} will likely hold even if we are testing weaker nulls with our algorithm. Second, a price we pay for using $|Y_i|$ to achieve honest inference is that we collapse the sign of the effect and therefore, CART treats subgroups with positive or negative effects equally. This is potentially problematic in settings where two different covariate values lead to identical effects (in magnitude), but different in signs; see \cite{hsu2013effect} for additional discussions and Web Appendix D for a numerical illustration. For our Medicaid example, if there is a partition of the covariates that leads to two identical H-CACEs in magnitude, but different in signs, our algorithm may not be able to detect the two subgroups. But, since using Medicaid is unlikely to be harmful, we don't believe this will be a significant concern in our example, especially compared to the alternatives of not obtaining honest inference. Third, Proposition \ref{algorithmprop} does not describe the algorithm's statistical power to detect effect heterogeneity. The next section uses a simulation study to addresses power and other factors influencing discovery of H-CACEs.

\section{Simulations}

We conduct a simulation study to measure the performance of the proposed algorithm in two ways: (1) statistical power to test H-CACEs and (2) recovery of true effect modifier variables. Throughout the simulation study, we vary the the compliance rate because prior works have shown that performance of IV methods depends heavily on the compliance rate, or more generally on the instrument's association to the treatment (i.e. instrument strength). In particular, problems can arise when the compliance rate is low; see \cite{staiger1994instrumental}, \citet*{stock2002survey}, and references therein for more details. 

Following \cite{hsu2015strong}, each simulation setting sets the potential outcomes $r_{0ij}^{(d_{0ij})}$ and $r_{1ij}^{(d_{1ij})}$, potential treatments $d_{0ij}$ and $d_{1ij}$, and covariates $\mathbf{X}_{ij}$ of each unit $j$ within each of the $I= 2000$ pairs. There are six pre-instrument covariates, each generated from independent Bernoulli trials with $0.5$ probability of success. At most two covariates, $x_1$ and $x_2$, modify the treatment effect. That is, H-CACEs defined by $\lambda(x_1,\dots,x_6)$ in equation \eqref{hcateEqn} depend on at most two covariates, $x_1$ and $x_2$. Also, because both $x_1$ and $x_2$ are binary, there are at most four different H-CACEs defined by different combinations of binary variables $\lambda_{00}$, $\lambda_{01}$, $\lambda_{10}$, and $\lambda_{11}$; for notational simplicity, we use $\lambda_{x_1x_2}$ to represent equation \eqref{hcateEqn}. Similar to the design of the OHIE, the data is generated under the assumption of one-sided compliance. This means that for every unit, the potential treatment having not received the instrument is 0, $d_{0ij}=0$. The potential treatment having received the instrument, $d_{1ij}$, is then a Bernoulli trial with success rate $\pi$; $\pi$ is also the compliance rate. In Web Appendix C, we consider the setting in which the compliance rate may depend on $x_1$ and $x_2$, say via $\pi_{x_1x_2}$. Finally, the potential outcomes having not received the instrument $r_{0ij}^{(d_{0ij})}$ are from a standard normal distribution $r_{0ij}^{(d_{0ij})} \sim N(0,1)$, and the potential outcomes having received the instrument $r_{1ij}^{(d_{1ij})}$ are a function of the H-CACE $r_{1ij}^{(d_{1ij})} = r_{0ij}^{(d_{0ij})} + d_{1ij} \lambda_{x_1x_2}$. Once all the potential treatment and outcomes are generated, the observed treatment and outcome are determined based on the value of the instrument and SUTVA. 
Finally, the regression tree in Algorithm \ref{genprocedure} is estimated in R using the package \emph{rpart}, version 4.1-15 \citep*{therneau2015package}. Unless specified otherwise, we use a complexity parameter of 0.005 (half of the default setting) and use defaults for the rest of \emph{rpart}'s parameters. Our proposed method is referred to as ``H-CACE" in the results below.

For comparison, we also apply a recent method by \cite{bargagli2019heterogeneous} to discover and test H-CACEs. Briefly, their method, which we refer to as ``BCF-IV" in the results below, utilizes modern tree-based methods \citep{athey2015machine,hahn2017bayesian} to estimate heterogeneous intent-to-treat (ITT) effects and suggests different sub-populations of interest; we remark that unlike our proposal, their method does not use matching and uses the original, untransformed $R_{ij}$ inside the tree fitting step. Then, for each sub-population, the method estimates and tests its H-CACE using the two-stage least square estimator. We use the \emph{bcf\_iv} function available on the authors' Github repository and use the default parameters of \emph{rpart} and \emph {bcf} \citep{hahn2017bayesian}.

\subsection{Statistical Power} \label{sec:power}

To measure a method's statistical power, we divide the number of false null hypotheses rejected by the total number of false null hypotheses suggested by the method. We refer to this rate as the true discovery rate and is a common measure of statistical power in multiple testing settings.

We compute the true discovery rate at varying levels of instrument strength and four heterogeneous treatment settings: (a) No Heterogeneity, (b) Slight Heterogeneity, (c) Strong Heterogeneity, and (d) Complex Heterogeneity. In setting (a), there are no effect modifiers resulting in one subgroup with equal treatment effects, $\lambda_{00} = \lambda_{01} = \lambda_{10}=\lambda_{11} = 0.5$. In setting (b), there is one effect modifier $x_1$ resulting in two subgroups with similar but different treatment effects, $\lambda_{00}=\lambda_{01} = 0.7$ and $\lambda_{10}=\lambda_{11} = 0.3$. In setting (c), there is one effect modifier $x_1$ resulting in two subgroups with dissimilar treatment effects, $\lambda_{00}=\lambda_{01} = 0.9$ and $\lambda_{10}=\lambda_{11} = 0.1$. And, in setting (d),  there are two effect modifiers $x_1$ and $x_2$ resulting in three subgroups, one with a strong effect, two with no effects, and the last group with the average effect, $\lambda_{00}=1.5$, $\lambda_{01}=\lambda_{10}= 0$ and $\lambda_{11} = 0.5$. In all four settings, the overall complier average causal effect is $\lambda=0.5$. 

We repeat the simulation 1000 times for each treatment heterogeneity and instrument strength combination.  We remark that the null hypothesis is that of no treatment effect (i.e. $\lambda_0 = 0$) and, since all treatment heterogeneity settings have an overall complier average treatment effect of $\lambda =0.5$, only the hypotheses consisting of pairs with $\lambda_{x_1x_2} =0$ are true null hypotheses.

\begin{figure}[htbp]
    \centerline{\includegraphics[height=8cm,width=16cm, keepaspectratio]{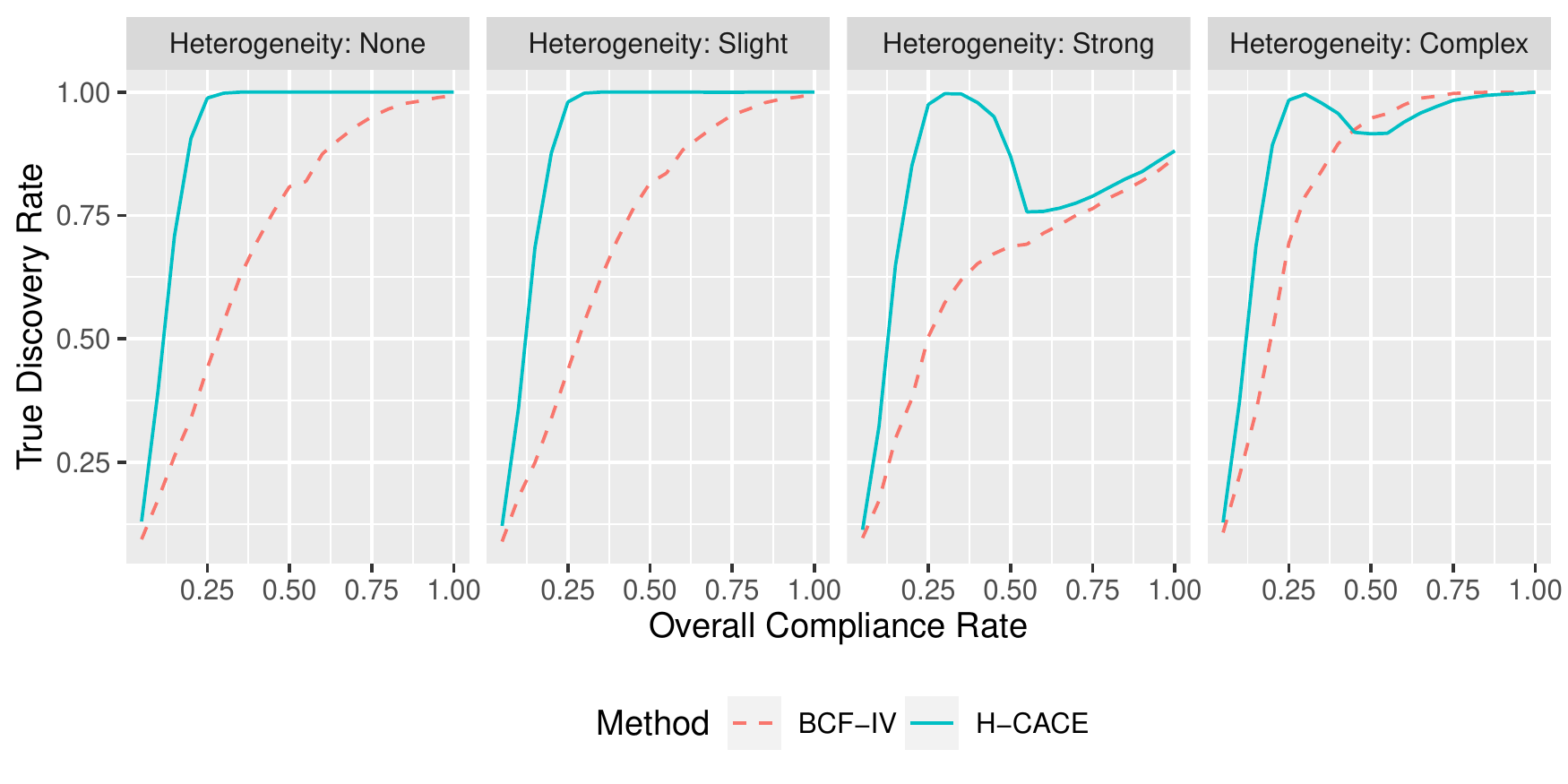}}
    \caption{True discovery rate as a function of the compliance rate and heterogeneity settings. The dashed and solid lines denote the BCF-IV procedure and our proposed algorithm, respectively.}
    \label{fig:hcaceBCFtdr}
\end{figure}

Figure \ref{fig:hcaceBCFtdr} shows the true discovery rate under four treatment heterogeneity settings. We see that as the compliance rate (i.e. instrument strength) increases, the true discovery rate of our method grows across all settings. In particular, our approach has the best power in the region where the compliance rate is low, roughly under 40\%.  Even when the compliance rate is high, we see that BCF-IV generally has lower power than our method across different heterogeneity settings.

We also take a moment to explain a counter-intuitive dip in our method's true discovery rate under the strong and complex heterogeneity settings in Figure \ref{fig:hcaceBCFtdr}. Briefly, this drop in the true discovery rate is due to the formation of leaves with smaller treatment effects.  As the compliance rate becomes large, these small effects begin to be get suggested by CART. But, the power to reject the null in favor of these small effects are small and the overall true discovery rate dips briefly. However, as the compliance rate reaches one, we see the true discovery rate of our method begin to climb again. Web Appendix E contains additional details surrounding this phenomena.

\subsection{False Positive Rate and F-Score} \label{sec:varsel}

We also assess our algorithm's ability to select variables among $\mathbf{X}_{ij}$ that are true effect modifiers. Specifically, we define a true effect modifier as a variable among $\mathbf{X}_{ij}$ where the tree splits on the variable and rejects one of the hypotheses of the split's children. A false effect modifier is a variable $X_{ij}$ where the tree either does not split on the variable or fails to reject on one of the hypotheses defined by the variable. Using this definition, we use the F-score and the false positive rate (FPR) common in the classification literature to measure a method's ability to select true effect modifiers. The F-score is the harmonic mean of recall and precision, where precision is the number of true positives (i.e. correctly selected true effect modifiers) out of the positive predictions (i.e. variables selected to be effect modifiers) and recall is the number of true positives out of the true conditions (i.e. true effect modifiers). More precisely, the F-score can be rewritten in terms of true positives (TP), false positives (FP), and false negatives (FN). 
\[F = \frac{2(\text{precision} \cdot \text{recall})}{\text{precision} + \text{recall}} =  \frac{TP}{TP + 0.5(FP + FN)}
\]
\noindent The F-score ranges from zero to one with a value closer to one implying greater accuracy. The FPR is defined as the number of false positives (i.e. incorrectly selected as a true effect modifier) out of the negative conditions (i.e. false effect modifiers) and ranges from zero to one, with a value close to zero being preferred.

\begin{figure}[htbp]
    \centerline{\includegraphics[height=9cm,width=17cm, keepaspectratio]{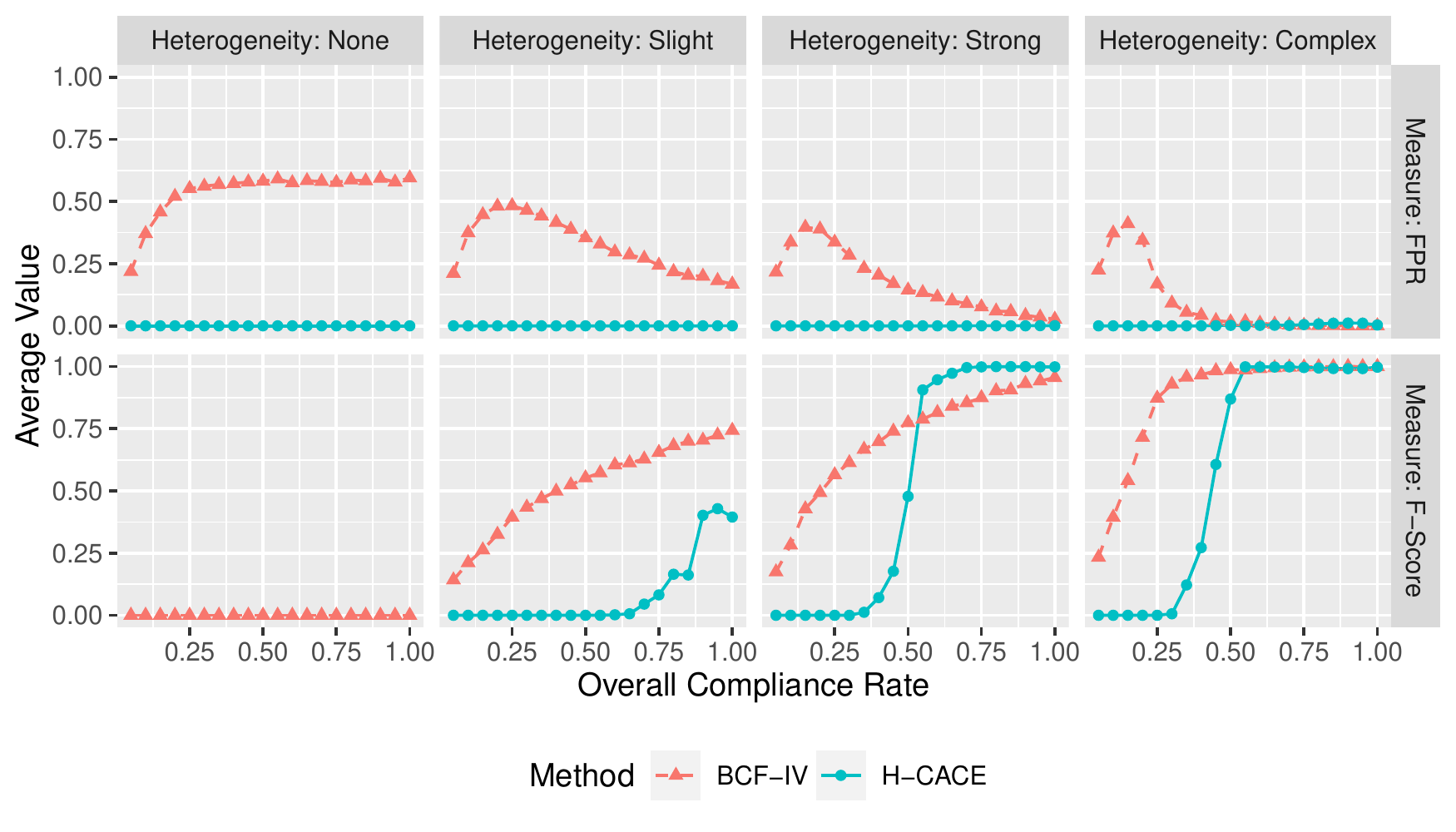}}
    \caption{F-score and false positive rate as a function of compliance rate and heterogeneity settings. The shapes and colors denote the two methods.}
    \label{fig:hcaceBCFem}
\end{figure}

We use the same four heterogeneity settings of (a) No Heterogeneity, (b) Slight Heterogeneity, (c) Strong Heterogeneity, and (d) Complex Heterogeneity. Figure \ref{fig:hcaceBCFem} shows the results of the F-score and FPR of our proposed algorithm and BCF-IV. Across four settings, our proposal has a false positive rate of nearly zero, never falsely declaring a variable to be a true effect modifier when it isn't in reality. In contrast, BCF-IV has a larger false positive rate, declaring effect modifiers to be true effect modifiers even if they are actually false effect modifiers. For example, in setting (a) without any effect modifiers, BCF-IV has a false positive rate hovering above 50\% whereas our method has a false positive rate of 0\%. In other words, BCF-IV falsely declared at least one of the six covariates as true effect modifiers roughly 50\% of the time whereas our method never declared any of the six covariates as true effect modifiers.

However, our algorithm's F-score is generally smaller than that from BCF-IV unless the compliance rate is high and the effect heterogeneity is strong. In particular, when the compliance rate is roughly under 50\% or if two subgroups have similar effect sizes, our method cannot select the true effect modifiers as well as BCF-IV. But, when the compliance rate is above 50\% and the effect heterogeneity is strong, our algorithm has similar F-scores as BCF-IV. Overall, the low F-score is a price that our algorithm pays for making sure that the FPR is small. In contrast, BCF-IV has a higher F-score, but pays a price with a high FPR.

In the supplementary materials Web Appendix C, D, and F, we conduct additional simulation studies where we (i) vary the compliance rate by covariates, (ii) allow H-CACEs to be equal in magnitude, but opposite in direction to measure the effect of using $|Y_i|$ in our algorithm, and (iii) demonstrate the two methods in a simulation that closely resembles the data from the OHIE. To summarize the results, for (i) and (iii), the story is very similar to what's presented here, where our method has high true discovery rate, low FPR and F-score compared to those from BCF-IV. For (ii), as expected, we find that our method has a low true discovery rate, FPR, and F-score. But, as soon as the magnitudes of the H-CACEs are dissimilar, our method returns to the case presented here.

\subsection{Takeaways from the Simulation Study} \label{sec:sum_sim}
Overall, the simulation study shows that our algorithm has large statistical power and low false positive rate across all settings. 
In contrast, the BCF-IV algorithm had low power and produced large FPRs, especially when no effect heterogeneity exists in the data; in other words, BCF-IV often falsely declared an effect modifier to be true effect modifiers. But, our algorithm generally had a low F-score compared to that from BCF-IV except in regimes where the effect heterogeneity is strong and the compliance rate is high. Clearly, no one method uniformly dominates the other in every data generating model under every metric of performance. Instead, we hope the simulation study here alerts investigators about the strengths and limitations of our algorithm compared to existing approaches. 

We also remark that the simulation results in Sections \ref{sec:power} and \ref{sec:varsel} do not necessarily contradict each other. Roughly speaking, the result in Section \ref{sec:power} concerns the ability for algorithms to have high \emph{statistical power} whereas the result in Section \ref{sec:varsel} concerns the ability for algorithms to \emph{select} variables. An algorithm like BCF-IV could liberally select many effect modifiers, generally leading to a high F-score and potentially a high FPR. But, the power to test the nulls suggested by the selected effect modifiers could be low since not only may some of these selected variables not be true effect modifiers, but also the selected variables will define many subgroups which likely contain few units to test the corresponding nulls. In contrast, an algorithm like ours could conservatively select effect modifiers, leading to a small F-score and low FPR. But, the power to test the nulls suggested by the selected variables could be high since most of the selected variables will be true effect modifiers. In short, our method is somewhat cautious, but certain whereas BCF-IV is optimistic, but somewhat error-prone.

\section{Analysis of the Oregon Health Insurance Experiment}

\subsection{Data Description}
We use our method to analyze the heterogeneous effects of Medicaid on the number of days an individual's physical or mental health prevented their usual activities in the past month. In brief, the OHIE collected administrative data on hospital discharges, credit reports, and mortality, survey data on health care utilization, financial strain, and overall health, and pre-randomization demographic data. There were 11,808 lottery winners and 11,933 lottery losers in the publicly available survey data for a total sample size of 23,741 individuals; see \citet{finkelstein2012oregon} for details.

We matched on the following demographic, pre-randomization variables recorded by \citet{finkelstein2012oregon}: sex, age, whether they preferred English materials when signing up for the lottery, whether they lived in a metropolitan statistical area (MSA), their education level (less than high school, high school diploma or General Educational Development (GED), vocational or 2-year degree, 4-year college degree or more), and self-identified race (as the individual reported in the survey). Since some of the covariates had missing data, namely self-identifying as Hispanic or Black and their level of education, we also matched on indicators of their missingness; see Section 9.4 of \citet{rosenbaum2010design} for details. We used the R package \emph{bigmatch}, version 0.6.1,  \citep{yu2019bigmatch} with an optimal caliper and a robust rank-based Mahalanobis distance to generate our optimal pair match. Figure \ref{fig:match_diagnostic} shows covariate balance before and after matching.

\begin{figure}[hbtp]
    \centerline{\includegraphics[width=14cm, height=14cm, keepaspectratio]{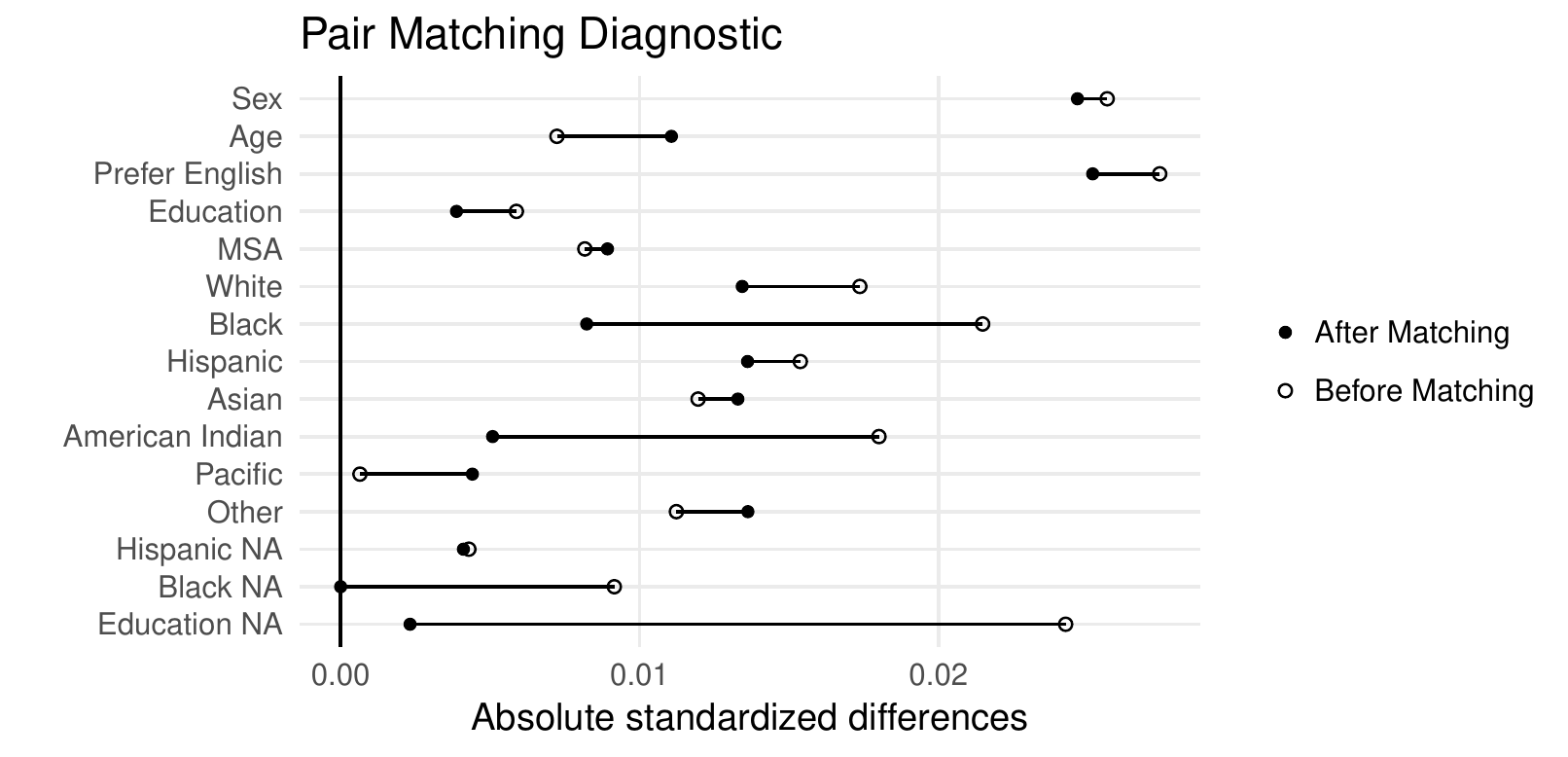}}
    \caption{Covariate balance as measured by  difference in means of the covariates between the treated and control groups, before and after matching.}
    \label{fig:match_diagnostic}
\end{figure}

For the majority of covariates, the matching algorithm did little to change the absolute standard differences between lottery winners and losers. This is not surprising given that the lottery was randomized. However, the indicator for missingness in education, self-identified American Indian, and Black were made to be more similar after matching. An absolute standardized difference of 0.25 is deemed acceptable \citep{rubin2001using, stuart2010matching}, which our covariates satisfied after matching.

\subsection{Instrument Validity} Before we present the results of our analysis using the proposed method, we discuss the plausibility of the lottery as an instrument. The lottery is randomized which ensures that the instrument is unrelated to unmeasured confounders and satisfying (A3). Winning the lottery, on average, increased enrollment of Medicaid by 30\% \citep{finkelstein2012oregon}, satisfying (A1). Assumption (A4) in the context of the OHIE states that there are no individuals who defy the lottery assignment to take (or not take) Medicaid if they lost (or won) the lottery. This is guaranteed by the design of the lottery, since an individual who lost the lottery cannot have access to Medicaid. However, we remark that \citet{finkelstein2012oregon} measured the treatment as whether or not an individual has ever had Medicaid during the study and a few individuals were already enrolled in Medicaid before the lottery winners were announced. Finally, assumption (A2) is the only assumption that could potentially be violated since individuals were not blind to their lottery results. This theoretically allowed lottery losers to seek other health insurance or lottery winners to make less healthy decisions since they're now able to be insured. These changes in an individual's behavior could affect his/her outcome regardless of his/her treatment and thus, may violate (A2).

\subsection{Analysis and Results} We run Algorithm \ref{genprocedure} and present the results in Figure \ref{fig:ohieCart}. We remark that we used \emph{rpart} in R with a complexity parameter of 0 and maximum depth of 4. The depth of the tree was chosen by forming trees of larger depth and then pruning back until a more interpretable tree was obtained. 
For each node of the CART, we tested whether or not there is an effect of enrolling in Medicaid $H_{0s}: \lambda_s = 0$. In Figure \ref{fig:ohieCart}, a solid lined box denotes a null hypothesis that was rejected and a dashed lined box denotes a null hypothesis that was retained, both by the closed testing procedure. Each node contains its estimated H-CACE $\hat{\lambda}_s$, 95\% confidence interval, the number of pairs $I_s$, and the estimated compliance rate $\hat{\pi}_s$. Here, a positive H-CACE implies a decrease in the number of days where the individual's physical and mental health prevented them from their usual activities, and a negative value implies an increase; in short, positive effects are beneficial to individuals. Also, some nodes imply a significant effect of Medicaid at level 0.05, but are enclosed in a dashed lined box. This is due to the closed testing procedure; an intersection of hypotheses containing the node in question was not rejected, and so any hypotheses in this intersection could not be rejected.

From Figure \ref{fig:ohieCart}, we can see evidence of heterogeneous treatment effects among the complier population. Specifically, Medicaid had a strong effect (1) among complying non-Asian men over the age of 36 and who prefer English, as well as (2) complying individuals younger than 36, who prefer English, and does not have more than a high school diploma or GED. Interestingly, among non-Asians over the age of 36 and who prefer English, females did not benefit from Medicaid as much as males even though the female subgroup was larger than the male subgroup and the compliance rates between the two subgroups were similar. 

More generally, while there is some variation in the compliance rates between groups, most of them are minor and hover between 25\% to 30\%. The minor variation suggests that while some subgroups are more likely to be compliers than others, most of the effect heterogeneity is likely driven by the variation in how the treatment differentially changes the response across subgroups; a bit more formally, most of the effect heterogeneity is likely arising from the numerator of the H-CACE rather than the denominator of the H-CACE. 

\begin{figure}[hbtp]
    \centerline{\includegraphics[width=14cm, keepaspectratio]{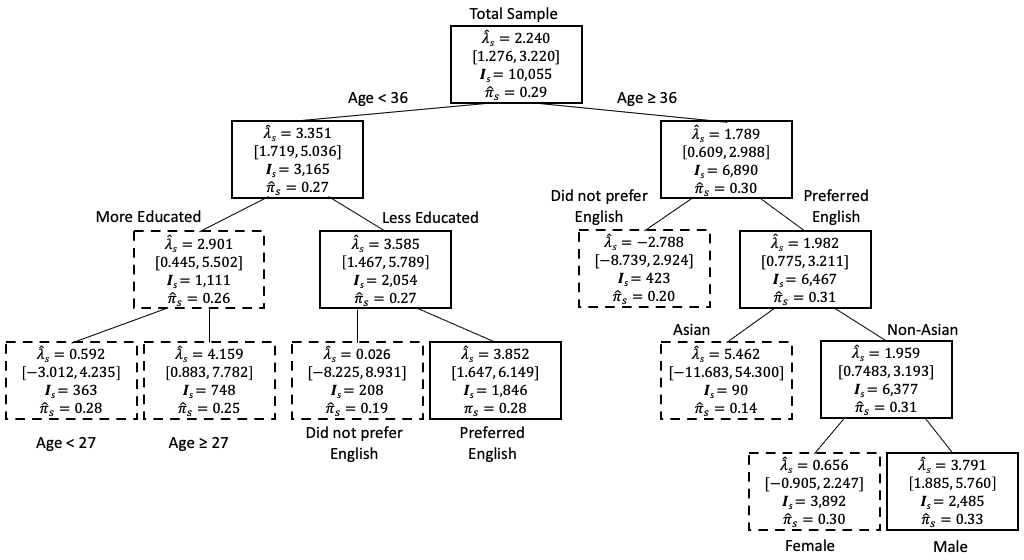}}
    \caption{Results of our proposed method on the effect of enrolling in Medicaid on the number of days physical or mental health did not prevent usual activities. Here, less educated refers to pairs with at most a high school diploma or GED and more educated refers to pairs with a higher education. Also, positive effects are beneficial to individuals. Solid lined boxes denote hypothesis tests that were rejected and dashed lined boxes denote hypotheses that were retained by closed testing. Within each box, the subgroup-specific estimated H-CACE $\hat{\lambda}_s$, its 95\% confidence interval, sample size of pairs $I_s$, and the estimated compliance rate $\hat{\pi}_s$ are provided. }
    \label{fig:ohieCart}
\end{figure}

\subsection{An Example of Closed Testing} \label{sec:app_closedtest} To better illustrate the closed testing portion of Algorithm \ref{genprocedure}, we walk through an example of the testing procedure based on the OHIE. As seen in Figure \ref{fig:ohieCart}, CART produced a tree with $G=8$ leaves. Now, consider testing whether there is evidence of a heterogeneous effect of Medicaid for young individuals who prefer English and have at most a high school diploma or GED, i.e. node $s_4$ in Figure \ref{fig:ctex} and $\mathcal{L} = \{4\}$ using Algorithm \ref{genprocedure}'s notation. The null hypothesis of interest would be $H_{0s_4}$, for all $j=1,2$ and $i \in s_4$. We then test and reject all of the hypothesis tests containing group $s_4$. For example,  we need to test the null hypothesis concerning the ancestor of $s_4$, say the subgroup of individuals who are younger than 36 and have at most a high school diploma or GED denoted as $\mathcal{L}'=\{3, 4\}$; see part (a) of Figure  \ref{fig:ctex}. Additionally, we need to test and reject all of the supersets containing $\mathcal{L}'$, which include but are not limited to the overall set $\{1,\dots, 8\}$, $\{1, \dots, 4\}$, and $\{3, 4, 6\}$. If every superset hypothesis and $H_{0s_4}$ are rejected at level $\alpha$, we can declare the effect in node $s_4$ to be significant and, by Proposition \ref{algorithmprop}, the familywise error rate is controlled at $\alpha$. Repeating this process for every node in the tree will give the results in Figure \ref{fig:ohieCart}.

\begin{figure}
\begin{subfigure}{.49\textwidth}
  \centering
  \includegraphics[width=.9\linewidth]{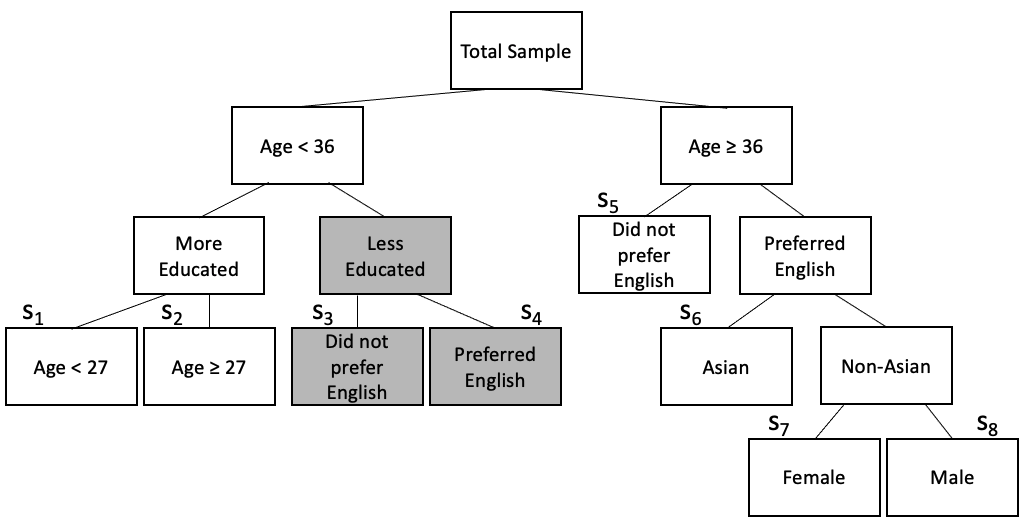}  
  \caption{Subgroup $\mathcal{L}'=\{3,4\}$}
  \label{fig:ct1}
\end{subfigure}
\begin{subfigure}{.49\textwidth}
  \centering
  \includegraphics[width=.9\linewidth]{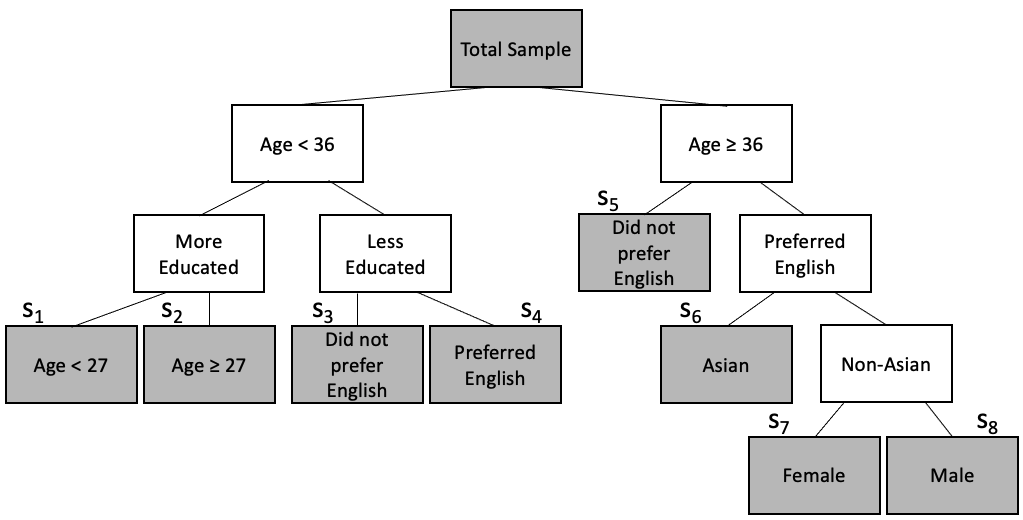}  
  \caption{Overall set $\{1, \dots, 8\}$}
  \label{fig:ct2}
\end{subfigure}


\begin{subfigure}{.49\textwidth}
  \centering
  \includegraphics[width=.9\linewidth]{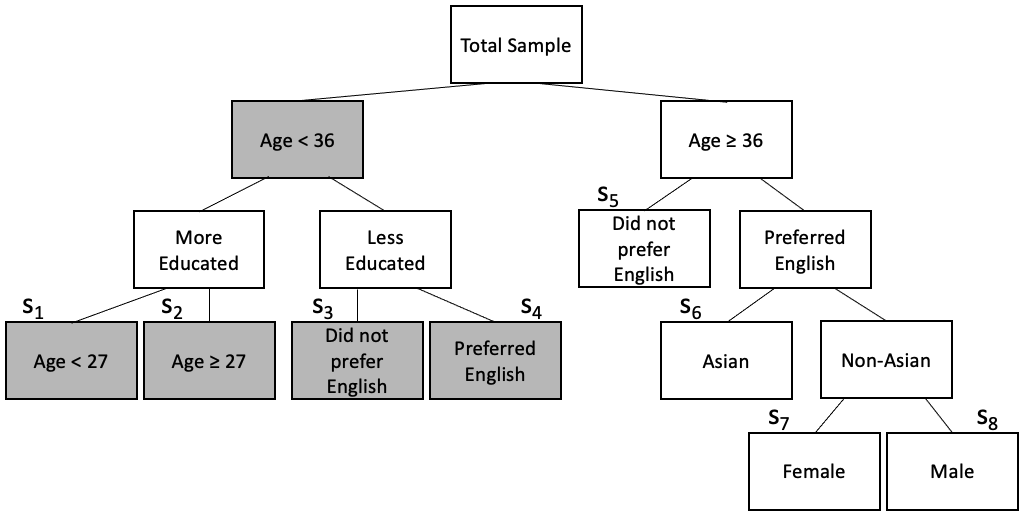}  
  \caption{Set $\{1, \dots, 4\}$}
  \label{fig:ct3}
\end{subfigure}
\begin{subfigure}{.49\textwidth}
  \centering
  \includegraphics[width=.9\linewidth]{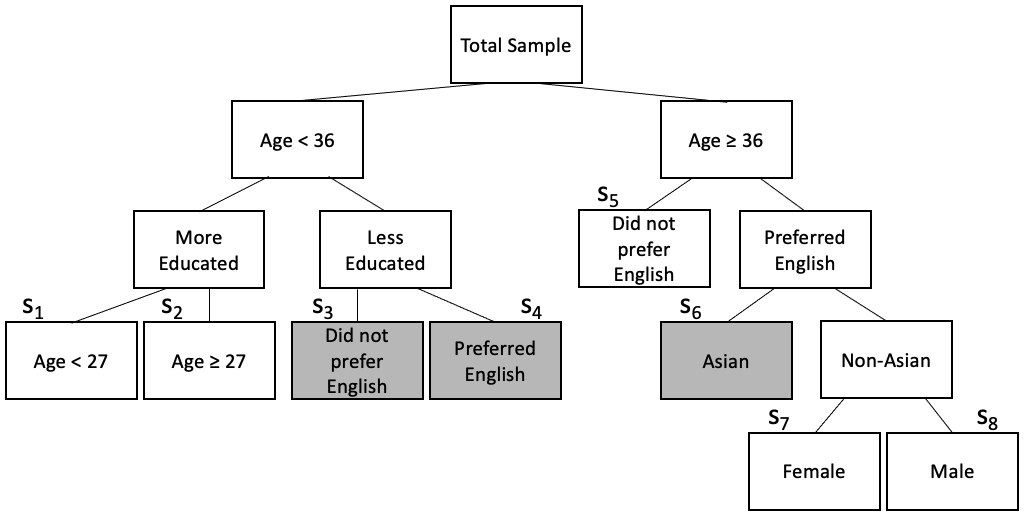}  
  \caption{Set $\{3, 4, 6\}$}
  \label{fig:ct4}
\end{subfigure}
\caption{Illustration of closed testing to test the null hypothesis $H_{0s_4}$ for all $j=1,2$ and $i \in s_4$. Each subplot highlights subsets required to be tested and rejected as part of closed testing.}
\label{fig:ctex}
\end{figure}

\section{Discussion}

In this paper, we propose a method based on matching to detect effect heterogeneity using an instrument. Under the usual IV assumptions, our method discovers and tests heterogeneity in the complier average treatment effect by combining matching, CART, and closed testing, all without the need to do sample splitting. The latter is achieved by taking the absolute value of the adjusted pairwise differences to conceal the instrument assignment and this allows our proposed method to control the familywise error rate. We also conducted a simulation study to examine the performance of our method and compared it to a recent method referred to as BCF-IV. Our method was then used to study the effect of Medicaid on the number of days an individual's physical or mental health did not prevent their usual activities where we used the lottery selection as an instrument. We found that Medicaid benefitted complying, older, non-Asian men who selected English materials at lottery sign-up and for complying, younger, less educated individuals who selected English materials at lottery sign-up. 

We conclude by making some recommendations about how to properly use our algorithm in practice, especially in light of existing approaches. First, as explained in the introduction, when there is noncompliance, exploring heterogeneity in the ITT alone with existing methods may provide an incomplete picture of the nature of the treatment effect. Relatedly, in settings where unmeasured confounding is unavoidable, our method based on an instrument is a promising way to discover and test effect heterogeneity. 

Second, as alluded to in Section \ref{sec:sum_sim}, the simulation results suggest that our algorithm tends to be conservative in discovering novel effect modifiers, reporting mostly true effect modifiers only if there is strong evidence for heterogeneity and minimizing selection of false effect modifiers. In other words, investigators can be reasonably confident that effect heterogeneity exists among the variables declared by our algorithm as effect modifiers. But, those variables that are not selected by our algorithm may also be true effect modifiers, but with slight effect heterogeneity. In such cases, investigators may need  additional samples to detect them using our method. In contrast, BCF-IV tends to be anti-conservative, reporting more effect modifiers, some of which may not actually be true effect modifiers. While this may be advantageous in situations where there is slight effect heterogeneity, investigators may not feel as confident about whether the detected effect heterogeneity truly exists.

Third, most recent approaches on effect heterogeneity, notably \cite{chernozhukov2018generic}, utilize sample splitting to achieve honest inference (i.e. type I error rate control) whereas our method uses absolute value of matched pairs to achieve it; note that both methods theoretically allow for a large class of machine learning methods to detect heterogeneous treatment effects, even though ours focused on CART for its simplicity and interpretability. While our method uses the full sample for both discovery and honest testing compared to those based on sample splitting, one of the caveats of our method is that our method may not be able to detect subgroups with identical effect sizes, but in opposite signs. 
Overall, every algorithm for effect heterogeneity carries some trade-offs and we urge investigators to understand their strengths and limitations to solidify and strengthen causal conclusions about effect heterogeneity in IV studies.

\newpage
\section*{Web Appendix A}
\subsection*{Honest Simultaneous Discovery and Inference}
\label{honestinference}
One advantage of our method is being able to use the entire data for discovering and testing effect modifiers. In order to simultaneously discover and draw inference on the sample, we use the absolute value of the pairwise differences $|Y|$ as the outcome of CART to obscure the sign of the difference in adjusted outcomes and preserve the original distribution of the instrumental variables (i.e. distribution based on assumption (A3)). We can then use this distribution to draw inference on our discovered potential effect modifiers. We study this phenomenon in two cases: (1) testing a single effect modifier (i.e. one hypothesis) and (2) testing multiple effect modifiers (i.e. multiple hypothesis).

In the first case, we are concerned about testing a single hypothesis and controlling the Type I error rate after discovering the hypothesis via CART. To investigate the effect of simultaneously discovering and drawing inference, we generate the potential outcomes with no treatment effect, $\lambda_{x_1x_2}=0$ for all $x_1, x_2$, and form potential effect modifiers for two different cases, (i) using $|Y|$ and (ii) $Y$ as the outcome for CART. The first leaf of each tree (or tree's root if no leaves are formed) is then used to test the null hypothesis of no treatment effect. As a result of conducting CART to form a hypothesis 2000 times, Figure \ref{fig:pvalHistAbsYvsY} shows the histogram of $p$-values when we use $|Y|$ as the outcome for CART versus $Y$ as the outcome for CART. Under $|Y|$, the $p$-values resemble a uniform distribution and hence, Type I error is controlled. However, under $Y$, the $p$-values are right-skewed implying that Type I error is inflated. In other words, the null hypothesis is rejected more frequently when using $Y$ as the outcome of CART, demonstrating the ``winner's curse" phenomena. Therefore, as predicted by Proposition \ref{algorithmprop}, using $|Y|$ as the outcome in CART prevents the contamination of the $\alpha$ level of the hypothesis test and allow for simultaneous discovery and inference.

\begin{figure}[htbp]
    \centering
    \includegraphics[width=15cm, keepaspectratio]{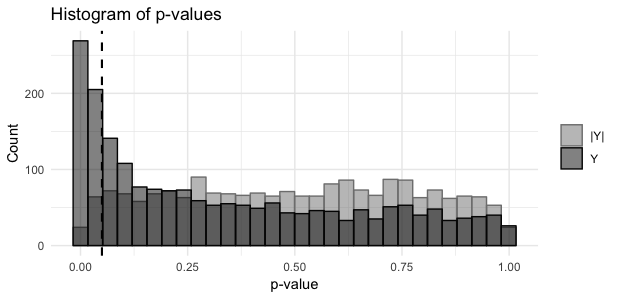}
    \caption{Histogram of $p$-values obtained from using $|Y|$ and $Y$ as the outcomes in CART in discovery of potential effect modifiers. The black dashed line denotes the alpha level of 0.05 of the hypothesis tests.}
    \label{fig:pvalHistAbsYvsY}
\end{figure}

In the second case, with multiple hypotheses, we are concerned with the average Type I error rate and strong control of the familywise error rate rate when testing multiple hypotheses. To evaluate strong control of the familywise error rate, we generate the data where there is an effect in some groups and no effect in others, $\lambda_{00}=2$, $\lambda_{01}=\lambda_{10}=0 = \lambda_{11} = 0$. Furthermore, for \emph{rpart}, we reduce the complexity parameter to 0.0001 to encourage more liberal splitting and set the max depth of the regression tree to be 4 to save on computational time. After data generation we randomly assign instrument values within pairs and split the data using $|Y|$ and $Y$ 2000 times. Each hypothesis is tested for whether or not there is a treatment effect $H_0: \lambda = 0$, so true hypotheses are hypotheses of groups of pairs generated with $\lambda_{01}=\lambda_{10}=\lambda_{11}=0$. The average Type I error rate is computed by taking the average of the proportion of false rejections from each of the 2000 simulated trees and the familywise error rate is computed by taking the average of any false rejections amongst the 2000 simulated trees.

\begin{table}
\centering
\begin{tabular}{|c|c|c|}
\hline
Simulation Setting & Mean Type I error rate & Familywise error rate \\
 \hline
 $|Y|$ & 0.0008 & 0.028\\
 \hline
 $Y$ & 0.0003 & 0.028\\
\hline
\end{tabular}
\caption{Results of simulations analyzing strong control of familywise error rate.}
\label{table:strongfwer}
\end{table}

The results of this simulation show that the average type I error rate and familywise error rate are below the $\alpha$ level of the hypothesis tests in both simulation settings $|Y|$ and $Y$. The average Type I error rate for $|Y|$ and $Y$ is 0.0008 and 0.0003, respectively. The familywise error rate for both $|Y|$ and $Y$ is 0.028 (See Table \ref{table:strongfwer}). This is surprising considering that closed testing requires that each hypothesis test be level $\alpha$ to strongly control the familywise error rate and using $Y$ as the outcome contaminates the test's level. Despite the theoretical underpinnings for this data generation process, closed testing seems to strongly control the familywise error rate regardless of whether or not the test's size is preserved by a technique such as taking the absolute value. Upon closer investigation, it seems that the trees formed in both $Y$ and $|Y|$ cases are the same at the upper levels of the tree, as there is a particularly strong signal for a certain group $\lambda_{00} = 2$. This then leads to the same hypotheses in both settings resulting in similar Type I error rates and familywise error rates. Overall, the simulation suggests that in the case where there is one very strong signal, the difference between using $|Y|$ and $Y$ is minor. But, we do stress familywise error control is only guaranteed for the $|Y|$ case.

\section*{Web Appendix B}
\setcounter{proposition}{0}
\begin{proposition}
[Familywise Error Rate Control of Algorithm \ref{genprocedure}] \label{algorithmpropproof}
Under the sharp null hypotheses $H_{0\mathcal{L}}$ in Algorithm \ref{genprocedure}, the conditional probability given $(\mathcal{F}, \mathcal{Z}, \mathcal{G})$ that the algorithm makes at least one false rejection of the set of hypotheses is at most $\alpha$.
\end{proposition}

\begin{proof}
 Define $h \subseteq \{1, \dots, I\}$ to be the union of all groups of pairs for which the null hypothesis is true; the groups of pairs of individuals which have an effect ratio of $\lambda_0$. In order to have a notion of type I error, some hypothesis or hypotheses must be true, so we assume that $h \neq \emptyset$ and that hypothesis $H_{0\mathcal{K}}: \lambda_s = \lambda_0$ for $s = \bigcup_{g \in \mathcal{K}} s_g$ is true. Note that by definition of $h$, in order for $H_{0\mathcal{K}}$ to be true, the groups in $\mathcal{K}$ are also contained in $h$, $s \subseteq h$. To make a type I error and reject hypothesis $H_{0\mathcal{K}}$, Algorithm \ref{detailedprocedure} must reject the intersection of all true hypotheses $H_{0\mathcal{T}}$, where $h=\bigcup_{g \in \mathcal{T}}$ and $\mathcal{K} \subseteq \mathcal{T}$. Yet, rejecting $H_{0\mathcal{T}}$ requires $\bigl \lvert \frac{T(\lambda_0)}{S(\lambda_0)} \bigr \rvert \geq z_{1-\alpha/2}$, where $P \left( \bigl \lvert \frac{T(\lambda_0)}{S(\lambda_0)} \bigr \rvert \geq z_{1-\alpha/2} \mid \mathcal{F}, \mathcal{Z}, \mathcal{G} \right) = \frac{\alpha}{2}$. Therefore, to make a type I error and reject $H_{0\mathcal{K}}$, one must reject $H_{0\mathcal{T}}$ which is a level $\alpha$ test.
\end{proof}

\setcounter{algocf}{0}

\IncMargin{1em}
\begin{algorithm}[ht]
\DontPrintSemicolon
\SetKwInOut{Input}{Given}\SetKwInOut{Output}{Output}
\SetKwBlock{NewBlock}{}{}
\Input{Observed outcome $R$, binary instrument $Z$, exposure $D$, covariates $X$, null value $\lambda_0$ for testing, and desired familywise error rate $\alpha$}
\BlankLine
\nl Pair match on observed covariates.\;
\nl Calculate absolute value of pairwise differences for each matched pair
\[\bigl\lvert Y_i \bigr\rvert = \bigl\lvert(Z_{i1} - Z_{i2})(R_{i1} - \lambda_0 D_{i1} - (R_{i2} - \lambda_0 D_{i2})) \bigr\rvert\]\;
\vspace{-4mm}
\nl Construct mutually exclusive and exhaustive grouping 
using CART. Here, CART takes $|Y_i|$ as the outcome and $\mathbf{X}_i$ from each matched pair as the predictors and outputs a partition of covariates, which we use to define $\mathcal{G}= \{s_1, \dots, s_G\}$ and consequently, H-CACEs. \;
\nl Run closed testing \citep{marcus1976closed} to test statistical significance of H-CACEs for every subset  $\mathcal{L} \subseteq \{1, \dots, G\}$ of $G$ groups where each subset defines the null hypothesis of the form $H_{0 \mathcal{L}}: r_{1ij}^{(d_{1ij})} - r_{0ij}^{(d_{0ij})} = \lambda_0 (d_{1ij} - d_{0ij})$ for all $g \in \mathcal{L}$.  \;
\Indp \vspace{-3mm}

\Indm\SetAlgoNoLine\NewBlock{
    \SetAlgoLined\For{$ \mathcal{L} \subseteq \{1, \dots, G \} $}{ 
        \If{$H_{0\mathcal{L}}$ \emph{has not been accepted}}{
        Calculate $T_{s}(\lambda_0)$ and $S_{s}(\lambda_0)$ for $s=\bigcup_{g \in \mathcal{L}} s_g$\;
        \If{$\bigl\lvert\frac{T_{s}(\lambda_0)}{S_{s}(\lambda_0)}\bigr\rvert \leq z_{1-\alpha/2}$}{
        Accept the null hypothesis $H_{0\mathcal{K}}: \lambda_{\mathcal{K}} = \lambda_0$ for all $\mathcal{K} \subseteq \mathcal{L} \subseteq \{1, \dots, \mathcal{G}\}$}
        \Else{Reject $H_{0 \mathcal{L}}$}
        }
    }
}
\Output{Estimated and inferential quantities for H-CACEs (e.g. effect size, confidence interval, $p$-value) and novel H-CACEs from closed testing.}
\caption{Proposed method to discover and test effect heterogeneity in IV with matching}\label{detailedprocedure}
\end{algorithm}\DecMargin{1em}

\section*{Web Appendix C}
\subsection*{Detecting H-CACE under Varying Compliance}

In section 3, the simulation settings assumed constant compliance rates across the groups. But it is possible that the compliance rates vary between subgroups. Therefore, we further consider four varying compliance rate settings as an extension of understanding the method's performance. These four different compliance settings are referred to as (a) Same, (b) Similar, (c) Different (1), and (d) Different (2) and are functions of the overall compliance rate $\pi$. Each are categorized based on the distance from the overall compliance rate. If the overall compliance is less than a half, $\pi \leq 0.5$, a group's compliance rate is $\pi_{x_1x_2} = \pi + c_{x_1x_2}\pi$, and if $\pi > 0.5$, a group's compliance rate is $\pi_{x_1x_2} = \pi + c_{x_1x_2}(1-\pi)$ for some constant $c_{x_1x_2} \in [0,1]$. The four settings are then defined by the constants $c_{x_1x_2}$; (a) Same compliance: $c_{00}=c_{01}=c_{10}=c_{11}=0$; (b) Similar compliance: $c_{00}=c_{01}= -0.1$ and $c_{10}=c_{11}=+0.1$; (c) Different (1) compliance: $c_{00}=c_{01}= -0.5$ and $c_{10}=c_{11}=+0.5$; and Different (2) compliance: $c_{00}=-0.3$, $c_{01}= -0.5$, $c_{10}=+0.1$, and $c_{11}=+0.7$. When combined with the treatment heterogeneity settings, we have a total of 16 possible settings of heterogeneity in H-CACE. We also remark that subgroups experiencing a larger H-CACE have a small compliance rate. Again, we compare our method to the BCF-IV method \citep{bargagli2019heterogeneous}.

Figure \ref{fig:HypAll} shows the true discovery rate of the four compliance types for each treatment heterogeneity setting. As in Section 3, this is a measure of the statistical power, where the true discovery rate is defined as the number of false hypotheses rejected out of all false hypotheses tested by closed testing. The different facets denote the treatment effect heterogeneity and the compliance heterogeneity. As the compliance rates get more different, we observe a reduction in the true discovery rate. This is most noticeable in the strong heterogeneity setting, where we see pronounced differences in the true discovery rate among the different compliance settings. In this setting, we also see a change in the true discovery rate between the more different (Different (1) and (2)) and more similar (Same and Similar) compliance groups; the more different compliance groups have an increased true discovery rate after an overall compliance rate of $\pi \approx 0.45$. This is due to the low compliance rate in the groups with stronger H-CACE in Different (1) and Different (2) settings, obscuring the signal for which CART uses to define subgroups. Therefore, the dip described in Section 3 and explained in Web Appendix E that is observed in the Strong Heterogeneous and Same Compliance settings (Figure \ref{fig:HypAll}) occurs at a high overall compliance rate for the Different (1) and Different (2) settings. In the heterogeneity settings, as the overall compliance rate grows, so too does the subgroup-specific compliance rates, and so the true discovery rates of the compliance settings converge. This is evidence that our method's true discovery rate relies on both the subgroup-specific size of H-CACEs and subgroup-specific compliance rates.

\begin{figure}[htbp]
    \centerline{\includegraphics[width=17cm, keepaspectratio]{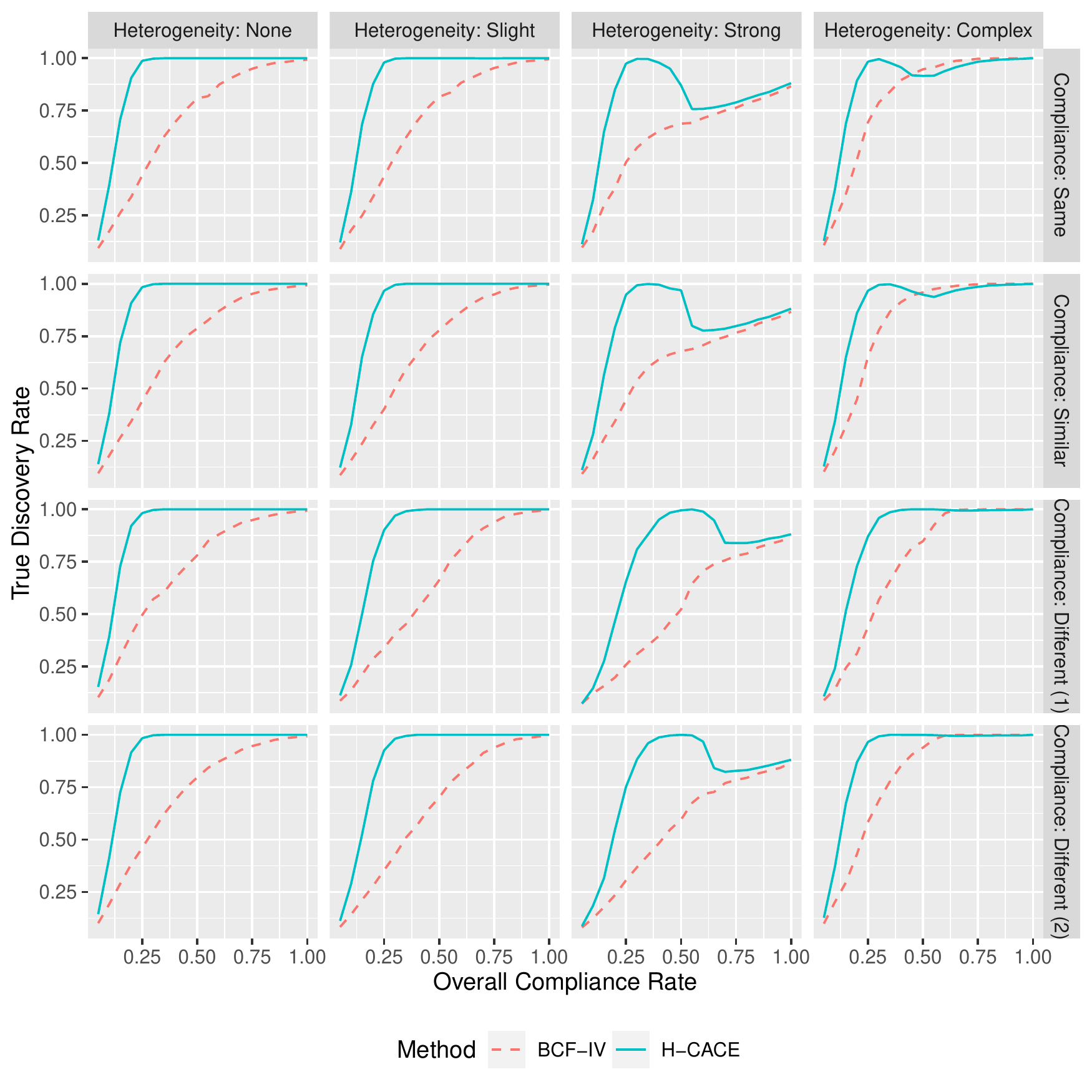}}
    \caption{True discovery rate as a function of overall compliance rate for the four treatment and four compliance heterogeneity settings. The color and line type denote the method, where the red dashed line denotes BCF-IV and the blue solid line denotes our method.}
    \label{fig:HypAll}
\end{figure}

Figure \ref{fig:emAll} shows the F-score and false positive rate (FPR) of the four compliance types for each treatment heterogeneity setting. As in Section 3, these are binary classification measures evaluating the algorithms' ability to detect true and false effect modification. For true positives (TP), false positives (FP), and false negatives (FN), the F-score is defined to be $F=\frac{TP}{TP + 0.5(FP + FN)}$, and the FPR is defined as the number of false positives out of the negative conditions. Similarly to Figure \ref{fig:HypAll}, as the compliance rates get more different, we observe a reduction in the performance of the algorithms. This is most noticeable in the Slight and Strong Heterogeneity settings for the BCF-IV algorithm and in the Strong and Complex Heterogeneity settings for our proposed algorithm. For the BCF-IV algorithm, the FPR begins to decline at a larger overall compliance rate in the more heterogeneous compliance settings, and the F-score of our algorithm climbs at a later overall compliance rate as well. As was the case with the true discovery rate, this is evidence that the FPR and F-score rely on both the H-CACE and compliance rate heterogeneity.

\begin{figure}[htbp]
    \centerline{\includegraphics[width=17cm, keepaspectratio]{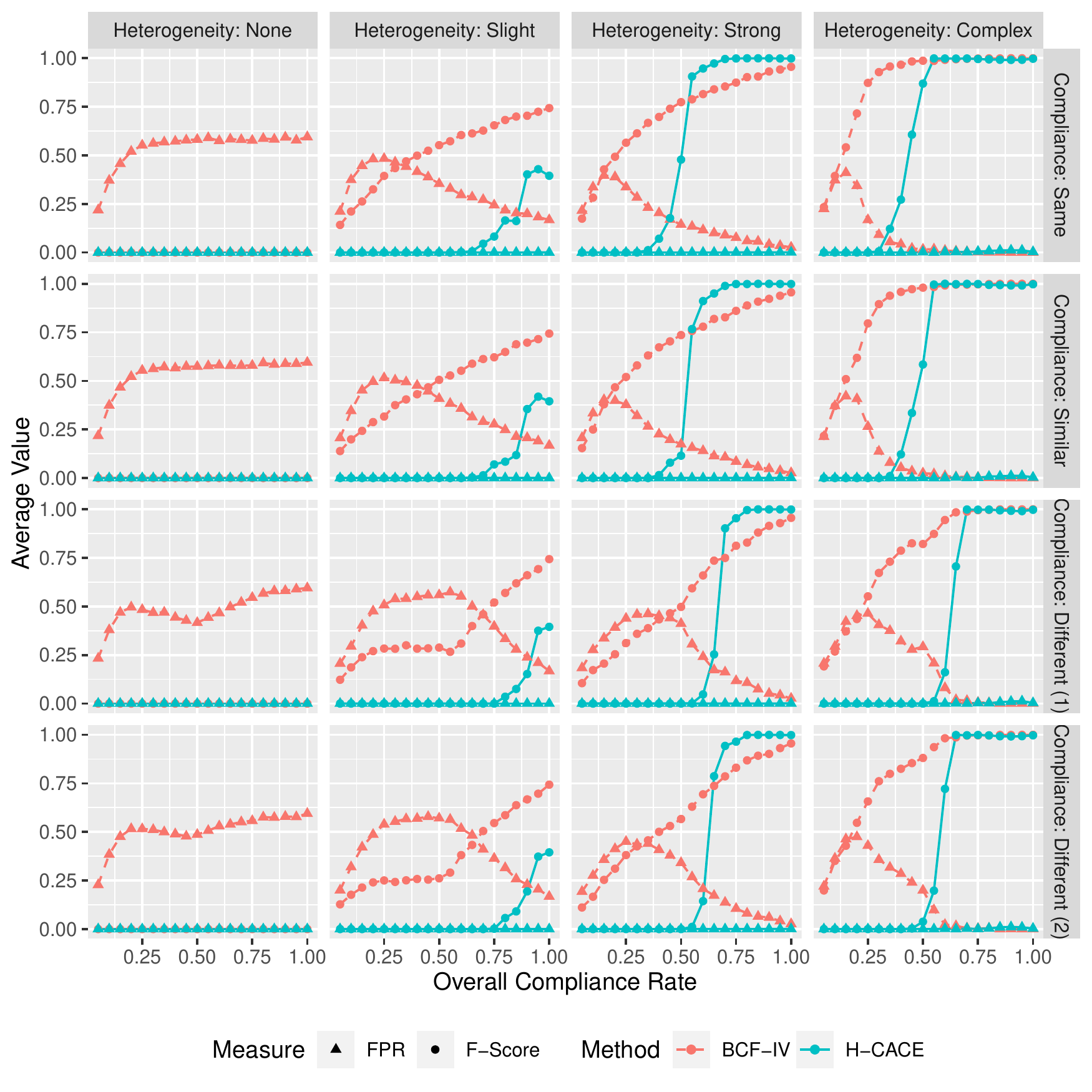}}
    \caption{F-score and false positive rates (FPR) as a function of overall compliance rate for the four treatment and four compliance heterogeneity settings. The shape of the points denote the measure, where a circle denotes the F-score and the triangle denotes the FPR. The color and line type denote the method, where the red dashed line denotes BCF-IV and the blue solid line denotes our method.}
    \label{fig:emAll}
\end{figure}

These simulations are evidence that our method's ability, as measured in statistical power and selection of effect modifiers relies on both the subgroup-specific size of H-CACEs and the subgroup-specific compliance rates.

\section*{Web Appendix D}
\subsection*{Testing for Equal, But Opposite Effects}

As it is possible that H-CACEs are equal in magnitude but opposite in direction, it is a question of how our method performs when we take the absolute value of the pairwise differences. To evaluate our method in this specific situation, we generate data as described in Section 3 but now considering only the heterogeneity setting $\lambda_{00}=0.3$, $\lambda_{01}=-0.3, \lambda_{10}= 0.7$ and $\lambda_{11} = -0.7$. Now, there are two effect modifiers $x_1$ and $x_2$ where $x_1$ changes the magnitude of the H-CACE and $x_2$ changes the direction. We compare our method to the BCF-IV method \citep{bargagli2019heterogeneous} as it does not transform the outcome and should still be able to detect the effect modification in this setting. 

\begin{figure}[htbp]
    \centerline{\includegraphics[width=17cm, keepaspectratio]{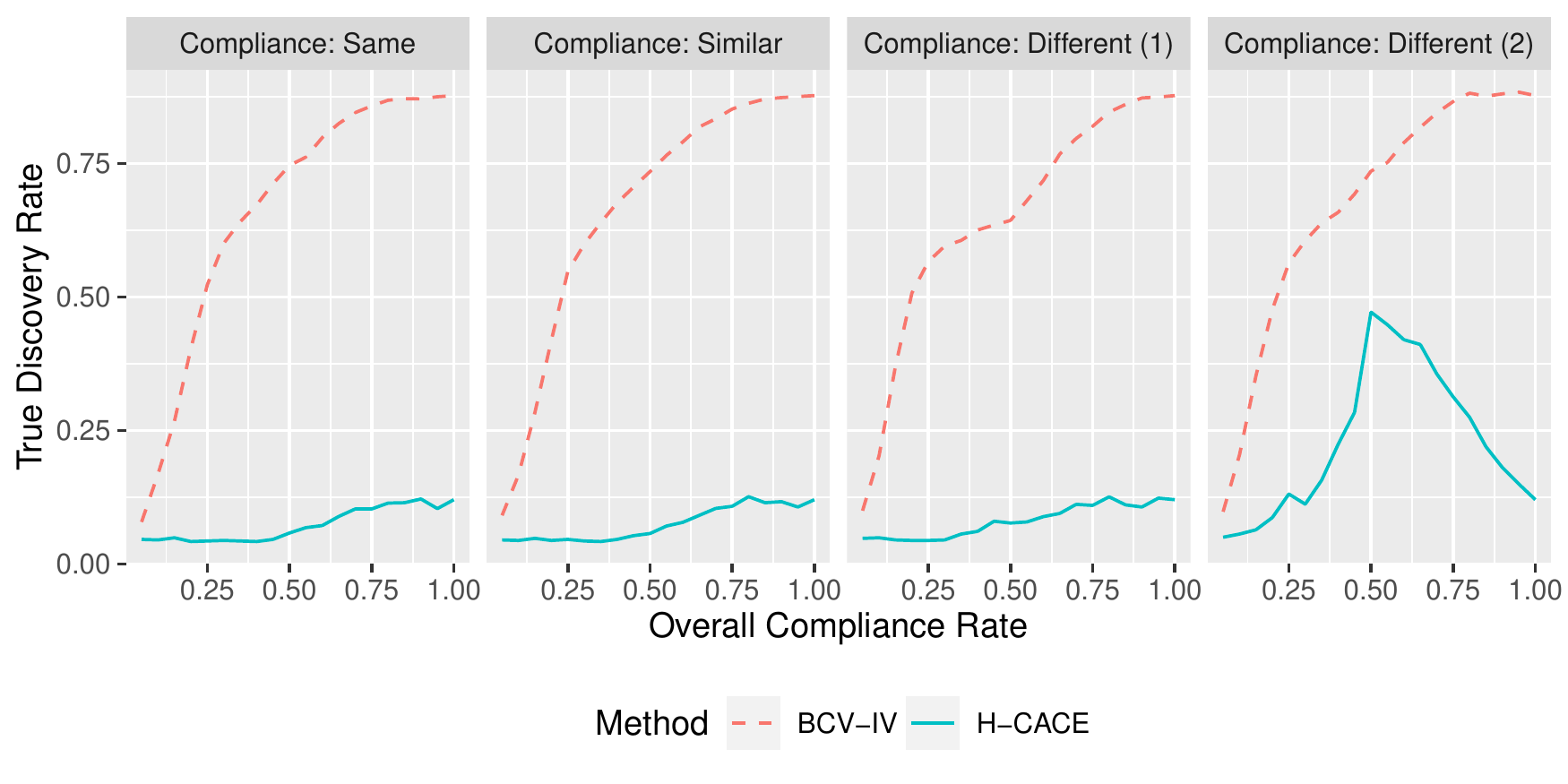}}
    \caption{True discovery rate as a function of overall compliance rate for four compliance heterogeneity settings. The linetype denotes the two methods, where a dashed line denotes the BCF-IV method and the solid line denotes our method.}
    \label{fig:Hypbcfhcace}
\end{figure}

\begin{figure}[htbp]
    \centerline{\includegraphics[width=17cm, keepaspectratio]{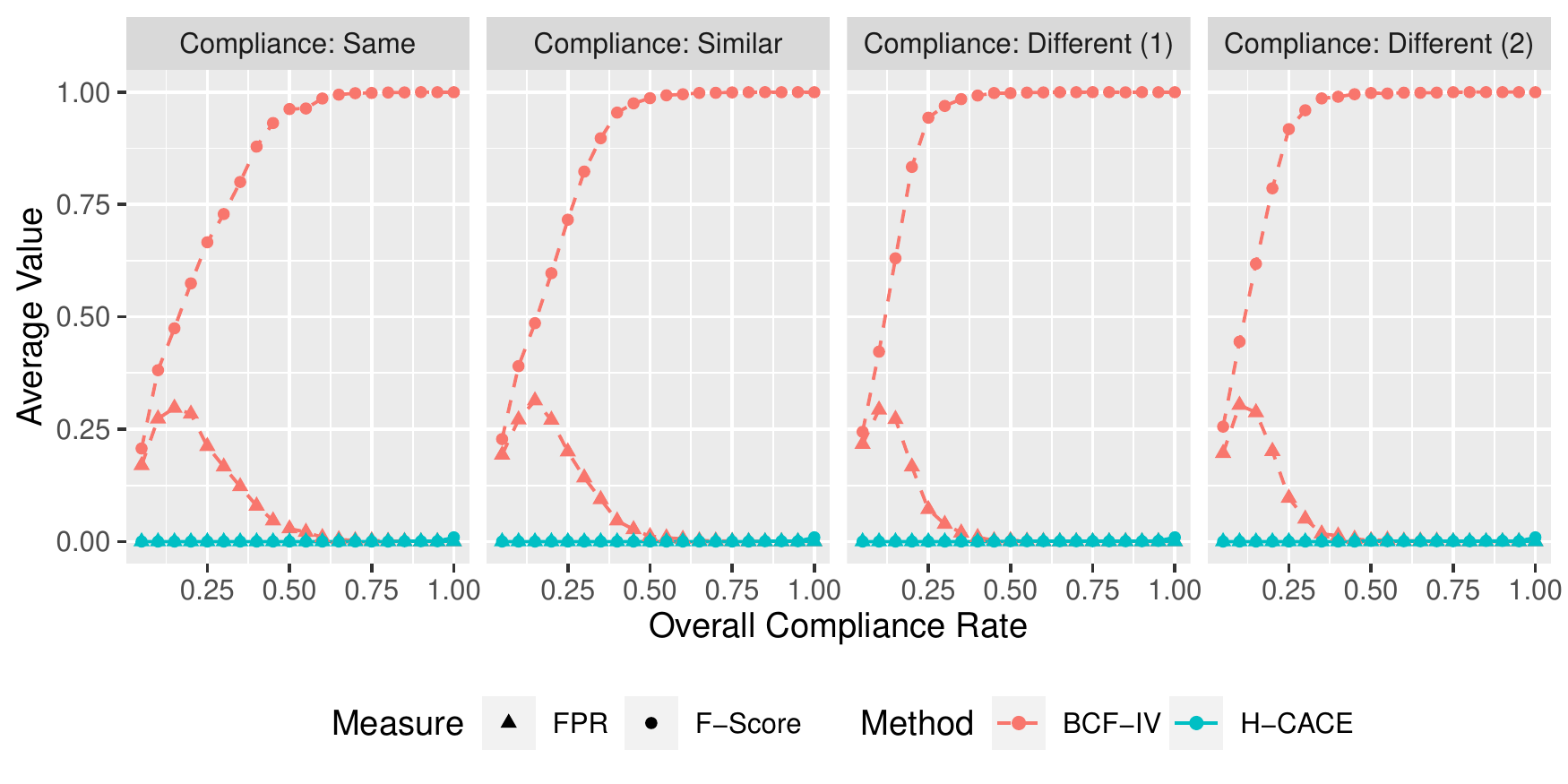}}
    \caption{False positive rate (FPR) and F-score as a function of overall compliance rate for the four compliance heterogeneity settings. The  color denotes the two methods, where a red line denotes the BCF-IV method and the blue line denotes our method. The point shapes denote the measure, where a triangle denotes FPR and a circle denotes F-score.}
    \label{fig:embcfhcace}
\end{figure}

In Figures \ref{fig:Hypbcfhcace} and \ref{fig:embcfhcace}, we see our method underperform in detecting heterogeneous treatment effects in this setting of equal but opposite effects, both in true discovery rate and F-score. The reason is two fold. First, in the transformation of the outcome of CART, the absolute value of the pairwise differences obscures the signal and hinders CART's ability to properly split. Second, in the case that CART does split, our closed testing procedure makes it challenging to reject the hypotheses suggested, since we must first reject the global hypothesis of $H_0:\lambda_0 = 0$. As we point out in Section 2.3, the CACE $\lambda$ is a weighted average of the H-CACEs in the global test, and it is therefore unlikely to reject as the average of the H-CACEs is 0. Together, CART's inability to split and the global hypothesis needing to be first rejected, our proposed method underperforms in the setting with equal but opposite effects. We therefore express caution in using our algorithm in the occasion an investigator believes the effect sizes are equal but opposite.

Interestingly, we see a spike of increased true discovery rate for both methods in the Different (2) compliance setting. In this setting, the compliance rates for the four groups are $\pi_{00}=0.7\pi$, $\pi_{01}=0.5\pi$, $\pi_{10}=1.1\pi$, and $\pi_{11}=1.7\pi$ when the overall compliance rate $\pi$ is less than or equal to 0.5, and $\pi_{00}=-0.3 + 1.3\pi$, $\pi_{01}=-0.5 + 1.5\pi$, $\pi_{10}=0.1 + 0.9\pi$, and $\pi_{11}=0.7 + 0.3\pi$ when $\pi>0.5$. Therefore, when $\pi$ is close to 0.5, the compliance rates are approach $\pi_{00}=0.35$, $\pi_{01}=0.25$, $\pi_{10}=0.55$, and $\pi_{11}=0.85$ for the four groups, changing the weights of the H-CACEs in the average for the global effect and shifting its value away from 0. As the overall compliance approaches 0 or 1, the heterogeneity in the compliance of the four groups reduces so the average of the H-CACEs approaches 0. This all improves our statistical test's ability to reject the global hypothesis and increase the true discovery rate. We do not see a large improvement however, as CART still fails to split. We note that in the case that the magnitudes are unequal, then our method will return to the performance demonstrated in Section 3, as CART will have signal to split on and our global hypothesis will more easily be correctly rejected. This can be seen in an example that the heterogeneity has the form $\lambda_{00}=\lambda_{01} = 0.9$, and the dotted line represents pairs from $\lambda_{10}=\lambda_{11} = -0.1$. With the absolute value transformation of the pairwise differences, our algorithm would treat this setting as in the Strong Heterogeneity setting we simulate in Section 3.

\section*{Web Appendix E}
\subsection*{Counter-intuitive Dip in True Discovery Rate}

As mentioned in Section 3 and shown in Figure \ref{fig:strongcomplexMarg}, we observe a counter-intuitive dip in true discovery rate as compliance rate grows for our proposed method. To investigate this drop, we also plot in different line types the true discovery rate of single subgroups formed by CART. The dashed lines denote leaves containing pairs with a stronger treatment effect and the dotted lines denote leaves containing pairs with a weaker treatment effect. For example, in the strong heterogeneity setting, the dashed line represents pairs from $\lambda_{00}=\lambda_{01} = 0.9$, and the dotted line represents pairs from $\lambda_{10}=\lambda_{11} = 0.1$. For the complex heterogeneity setting, the dashed line denotes pairs generated by $\lambda_{00}=1.5$; dotted lines aren't shown because CART fails to form a group consisting of only pairs generated by $\lambda_{11}=0.5$. By comparing the curves, we see that as the compliance rate grows the drop in the true discovery rate is due to the formation of leaves with smaller treatment effects.  Because the compliance rate is large enough, these small effects are beginning to be detected by CART. But, the power to detect these effects are much smaller than the large effects, and so the overall true discovery rate, which is roughly the average of these two curves, dips briefly. However, As the compliance rate grows, we see the true discovery rate of our method begin to climb again, as more signal for the smaller H-CACE groups is gained.

\begin{figure}[htbp]
    \centerline{\includegraphics[height=7cm,width=13cm, keepaspectratio]{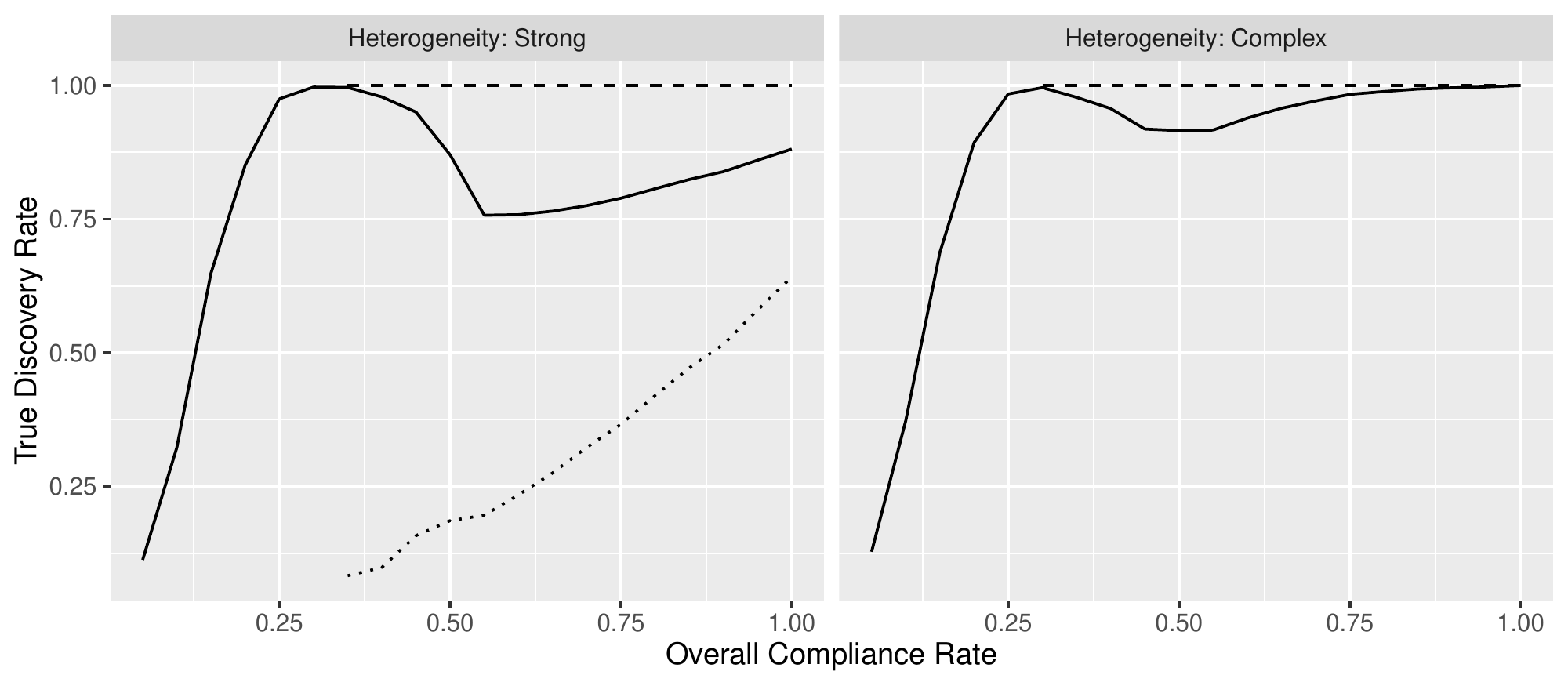}}
    \caption{True discovery rate as a function of overall compliance rate for the strong heterogeneity and complex heterogeneity settings. The line type denotes a single subgroup's treatment effect, where a dashed line denotes the stronger treatment effect and a dotted line the weaker treatment effect. }
    \label{fig:strongcomplexMarg}
\end{figure}

\section*{Web Appendix F}
\subsection*{Oregon Health Insurance Experiment Simulation}

Using the Oregon Health Insurance Experiment (OHIE) as a template for another simulation, we evaluate the performance of the proposed method under treatment magnitudes using the true discovery rate as a measure of the method's statistical power, and false positive rate (FPR) and F-score to measure the method's performance in determining effect modifiers. Using the matched pairs from the OHIE and their pre-instrument covariates $\mathbf{X_{ij}}$, we generate the potential treatments $d_{0ij}$ and $d_{1ij}$ and potential outcomes $r_{0ij}^{(d_{0ij})}$ and $r_{1ij}^{(d_{1ij})}$. Since the design of the OHIE ensures one-sided compliance, we generate the potential treatment without receiving the outcome to also satisfy one-sided compliance, $d_{0ij}=0$. To have similar compliance rates as observed in Section 4, the potential treatment having received the instrument is a Bernoulli trial with success rate $\pi(\mathbf{X_{ij}}) = 0.32 - 0.15(1-English) + 0.15(English \times Asian) - 0.05(Age < 36)$. Here, $English$ is a binary indicator where a value of $1$ denotes the individual's preference for English materials in signing up for the lottery, and $Asian$ is a binary indicator where a value of $1$ denotes the race reported in the survey as Asian. With this heterogeneous compliance rate, the overall compliance is the same as that estimated for the sample in Section 4, $\pi =0.29$. We then generate the potential outcomes as in Section 3, the potential outcomes having not received the instrument are from a standard normal distribution $r_{0ij}^{(d_{0ij})} \sim N(0,1)$, and the potential outcomes having received the instrument are a function of the H-CACE and effect modifiers detected in Section 4, $r_{1ij}^{(d_{1ij})} = r_{0ij}^{(d_{0ij})} + d_{1ij} \lambda(\mathbf{X_{ij}})$. The H-CACE $\lambda(\mathbf{X_{ij}})$ is a function of the effect modifiers $Age$, $Education$, $English$, $Asian$, and $Sex$ and the magnitude of the effects are defined in three settings:

\begin{itemize}
    \item[(i)] Small: $\lambda(\mathbf{X_{ij}}) = 0.25 + 4Age^{-1} + 0.1(1-Education) - 0.25(1-English) + 0.35Asian + 0.2(1-Sex)(Age \geq 36)$
    \item[(ii)] Moderate: $\lambda(\mathbf{X_{ij}}) = 0.5 + 8Age^{-1} + 0.2(1-Education) - 0.5(1-English) + 0.7Asian + 0.4(1-Sex)(Age \geq 36)$
    \item[(iii)] Large: $\lambda(\mathbf{X_{ij}}) = 1 + 16Age^{-1} + 0.4(1-Education) - 1(1-English) + 1.4Asian + 0.8(1-Sex)(Age \geq 36)$
\end{itemize}

\noindent As in Section 4, $Education$ is a binary variable where a value of 1 denotes a vocational degree, 2-year degree, 4-year college degree, or more.

For our proposed method, we use the R package \emph{rpart} with a complexity parameter of 0.001, max depth of 7, and minimum number of observations needed for a split to be 90. This allows CART to split on more variables than the number of effect modifiers while preventing the creation of nodes with very few observations. For BCF-IV, we use the default \emph{rpart} settings as in Section 3. The averages of the 1000 simulations at each treatment magnitude level are provided in Table \ref{table:ohiesimtab}. 

The results of this simulation are similar to those seen in Section 3, where our proposed method performs well in the true discovery rate, but performs poorly in the FPR and F-score when the compliance rate is low and the heterogeneity magnitudes are weak. This further aligns with our results in Section 3, as we found our method struggles selecting effect modifiers as measured by the F-score at compliance rates below 50\% and the compliance rate is 29\% in this setting. However, we see that our method improves in the F-score as the treatment magnitude increases. In contrast, BCF-IV outperforms our method in the F-score, but has an inflated FPR and a deflated true discovery rate. However, as the signal improves, BCF-IV's FPR reduces to a more ideal value and the true discovery rate grows. 

\begin{table}[h!]
    \centering
    \setlength{\tabcolsep}{0.5em} 
    {\renewcommand{\arraystretch}{1.3}
    \begin{tabular}{ |c||c|c c c| } 
    \hline
    Method & $\lambda(\mathbf{X_{ij}})$ & True Discovery Rate & FPR & F-Score \\
    \hline
    \multirow{3}{4em}{H-CACE} & Small & 1.00 & 0.00 & 0.00 \\ 
    & Moderate & 0.99 & 0.02 & 0.04 \\ 
    & Large & 0.99 & 0.06 & 0.76 \\ 
    \hline
    \multirow{3}{4em}{BCF-IV} & Small & 0.72 & 0.17 & 0.56 \\ 
    & Moderate & 0.83 & 0.05 & 0.74 \\ 
    & Large & 0.93 & 0.01 & 0.88 \\ 
    \hline
    \end{tabular}
    }
    \caption{Average true discovery rate, false positive rate (FPR), and F-score of the two methods at the different treatment magnitudes.}
    \label{table:ohiesimtab}
\end{table}

\newpage
\bibliography{bibliography.bib}
\end{document}